\pdfoutput=1
\documentclass[a4paper,USenglish,cleveref, autoref, thm-restate, nameinlink]{lipics-v2021} %
\hideLIPIcs
\nolinenumbers
\usepackage{environ}
\usepackage{complexity} %
\usepackage{dsfont}
\usepackage{pifont}
\usepackage{nicefrac}

\usepackage{tabularx}
\usepackage{booktabs}

\usepackage{etoolbox}

\newcommand{\appsymb}{$\star$}
\newcommand{\appref}[1]{{\hyperref[proof:#1]{\appsymb}}}

\newcommand{\appendixsection}[1]{%
}

\newcommand{\toappendix}[1]{%
    #1
}

\newcommand{\appendixproof}[2]{%
    #2
}

\usepackage{tikz}
\usetikzlibrary{matrix,fit,decorations.pathreplacing,shapes,patterns,patterns.meta,intersections,positioning,chains,calc,backgrounds}
\usetikzlibrary{decorations.markings} 
\usepackage{xifthen}

\definecolor{italyGreen}{RGB}{0, 146, 70}
\definecolor{italyRed}{RGB}{206, 43, 55}

\makeatletter
\newcommand{\problemtitle}[1]{\gdef\@problemtitle{#1}}%
\newcommand{\probleminput}[1]{\gdef\@probleminput{#1}}%
\newcommand{\problemquestion}[1]{\gdef\@problemquestion{#1}}%
\newcommand{\problemoutput}[1]{\gdef\@problemoutput{#1}}%
\NewEnviron{problem}{
  \problemtitle{}\probleminput{}\problemquestion{}%
  \BODY%
  \par\addvspace{.5\baselineskip}
  \noindent
  \begin{tabularx}{\textwidth}{@{\hspace{\parindent}} l X c}
    \multicolumn{2}{@{\hspace{\parindent}}l}{\@problemtitle} \\%
    \textbf{Input:} & \@probleminput \\%
    \textbf{Question:} & \@problemquestion%
  \end{tabularx}
  \par\addvspace{.5\baselineskip}
}
\NewEnviron{problemo}{
  \problemtitle{Critical Edges/nodes for Densest Subgraph problem}\probleminput{}\problemquestion{}%
  \BODY%
  \par\addvspace{.5\baselineskip}
  \noindent
  \begin{tabularx}{\textwidth}{@{\hspace{\parindent}} l X c}
    \multicolumn{2}{@{\hspace{\parindent}}l}{\@problemtitle} \\%
    \textbf{Input:} & \@probleminput \\%
    \textbf{Output:} & \@problemoutput%
  \end{tabularx}
  \par\addvspace{.5\baselineskip}
}
\makeatother

\newtheorem{rrule}[theorem]{Data Reduction Rule}

\crefname{rrule}{Data Reduction Rule}{Data Reduction Rules}

\newcommand{\NN}{\mathds{N}}

\newcommand{\densestEdgeDeletion}{\textsc{Bounded-Density Edge  Deletion}}
\newcommand{\densestEdgeDeletionMinK}{\textsc{Min Edge Deletion Bounded-Density}}
\newcommand{\densestEdgeDeletionMinRho}{\textsc{Min Density Edge Deletion}}
\newcommand{\densestVertexDeletion}{\textsc{Bounded-Density Vertex Deletion}}
\newcommand{\densestVertexDeletionMinK}{\textsc{Min Vertex Deletion Bounded-Density}}
\newcommand{\densestVertexDeletionMinRho}{\textsc{Min Density  Vertex Deletion}}

\DeclareMathOperator{\cost}{cost}
\DeclareMathOperator{\ob}{ob}
\DeclareMathOperator{\mad}{mad}

\title{Destroying Densest Subgraphs is Hard}

\author{Cristina Bazgan}{Universit\'e   Paris-Dauphine, PSL Research University, CNRS, UMR  7243, LAMSADE, Paris, France}{cristina.bazgan@dauphine.fr}{https://orcid.org/0000-0002-5460-6222}{}

\author{André {Nichterlein}}{Algorithmics and Computational Complexity, Technische Universität Berlin, Germany}{andre.nichterlein@tu-berlin.de}{https://orcid.org/0000-0001-7451-9401}{}

\author{Sofia {Vazquez Alferez}}{Universit\'e   Paris-Dauphine, PSL Research University, CNRS, UMR  7243, LAMSADE, Paris, France}{sofia.vazquez-alferez@lamsade.dauphine.fr}{https://orcid.org/0000-0002-1541-8683}{}

\authorrunning{C.\ Bazgan and A.\ Nichterlein and S.\ Vazquez Alferez}

\Copyright{Cristina Bazgan, André Nichterlein and Sofia Vazquez Alferez} %



\relatedversion{} %
\begin{document}

\maketitle

\begin{abstract}
	We analyze the computational complexity of the following computational problems called \densestEdgeDeletion{} and \densestVertexDeletion{}:
	Given a graph~$G$, a budget~$k$ and a target density~$\tau_\rho$, are there~$k$ edges ($k$ vertices) whose removal from~$G$ results in a graph where the densest subgraph has density at most~$\tau_\rho$?
	Here, the density of a graph is the number of its edges divided by the number of its vertices.
	We prove that both problems are polynomial-time solvable on trees and cliques but are NP-complete on planar bipartite graphs and split graphs.
	From a parameterized point of view, we show that both problems are fixed-parameter tractable with respect to the vertex cover number but W[1]-hard with respect to the solution size.
	Furthermore, we prove that \densestEdgeDeletion{} is W[1]-hard with respect to the feedback edge number, demonstrating that the problem remains hard on very sparse graphs.
\end{abstract}

\newpage

\section{Introduction}

Finding a densest subgraph is a central problem with applications ranging from social network analysis to bioinformatics to finance~\cite{LMFB23}.
There is a rich literature on this topic with the first polynomial-time algorithms given more than 40 years ago~\cite{G1984,PQ82}.
In this work, we study the robustness of densest subgraphs under perturbations of the input graph. %
More precisely, we study the (parameterized) complexity of the following two computational problems called \densestEdgeDeletion{} and \densestVertexDeletion{}: 
Given a graph~$G$, a budget~$k$ and a target density~$\tau_\rho$, the questions are whether there are~$k$ edges, respectively, $k$ vertices, whose removal from~$G$ results in a graph where the densest subgraph has density at most~$\tau_\rho$?
Here, the density~$\rho(G)$ of a graph~$G$ is defined as the ratio between its number~$m$ of edges and number~$n$ of vertices, that is, $\rho(G)=m/n$, which is equal to half the average degree of~$G$.
Thus, we contribute to the literature on graph modification problems with degree constraints~\cite{FNN16,MS12,Nic15}.
More broadly, our work fits into parameterized algorithmics on graph modification problems---a line of research with a plethora of results. See Crespelle et al.~\cite{CDFG23} for a recent survey focusing on edge modification problems.

Denote with $\rho^*(G)$ the density of a densest subgraph of~$G$.
Note that cycles have density exactly one and forests a density below one.
Thus, it is easy to see that a graph~$G$ is a forest if and only if~$\rho^*(G) < 1$.
Hence, our problems contain the \NP-hard \textsc{Feedback Vertex Set} and the polynomial-time solvable \textsc{Feedback Edge Set}\footnote{Given a graph~$G$ and an integer~$k$, \textsc{Feedback Vertex Set} (\textsc{Feedback Edge Set}) asks if there is a set of~$k$ vertices ($k$ edges) whose removal makes~$G$ acyclic.}, respectively, as special cases.
For target densities smaller than one not only are cycles to be destroyed, but also a bound on the size of the remaining connected components is implied. 
For example, for~$\tau_\rho = 2/3$ ($= 1/2$ or $=0$) each connected component in the resulting graph can have at most~$2$ edges ($1$ edge for~$\tau_\rho=1/2$ or $0$ edges for~$\tau_\rho=0$).
Consequently, \densestVertexDeletion{} generalizes \textsc{Dissociation Set}
($\tau_\rho = 1/2$) and \textsc{Vertex Cover} ($\tau_\rho = 0$).
\densestEdgeDeletion{} generalizes \textsc{Maximum Cardinality Matching} ($\tau_\rho = 1/2$) and the NP-hard \textsc{Maximum $P_3$-packing} ($\tau_\rho = 2/3$) where the non-deleted edges form the matching and $P_3$-packing, respectively.

\subparagraph*{Our contributions.}
We refer to \cref{tab:overview} for an overview of our results.
\begin{table}
	\caption{Our results for \densestEdgeDeletion{} / \textsc{Vertex Deletion}.}
	\label{tab:overview}
	\begin{tabularx}{\textwidth}{llXX}
		\toprule
		& & \textsc{Edge Deletion} & \textsc{Vertex Deletion} \\
		\midrule
		& Trees & $O(n^3)$ (\cref{thm:EdgeDel-poly-on-trees}) & $O(n)$ (\cref{thm:vertexDel-Trees}) \\
		& Cliques & $O(n^2)$ (\cref{thm:EdgeDel-poly-on-cliques}) & $O(n^2)$ (trivial) \\
		& Split & NP-complete (\cref{thm:edgeDeletion-NP-hard-on-Split}) & NP-complete (\cref{thm:vertex-del-split-graph}) \\
		& Planar Bipartite & NP-complete (\cref{thm:hardness-densest-edge-bipartite}) & NP-complete (\cref{thm:VertexDeletion_NP_Hard-bipartite})\\
		\midrule
		& Solution Size~$k$ & W[1]-hard (\cref{thm:edge-del-w-hard}) & W[2]-hard (\cref{thm:vertex-del-w2})\\
		& Vertex Cover Number & FPT (\cref{thm:edge-deletion-FPT-vertex-cover-number}) & FPT (\cref{thm:VertexDeletion_FPT_VC})\\
		& Feedback Edge Number & W[1]-hard (\cref{thm:edge-del-w-hard}) & ? \\
		\bottomrule
	\end{tabularx}
\end{table}
Given the above connections to known computational problems, we start with the seemingly easier of the two problems: \densestEdgeDeletion{}.
We provide polynomial-time algorithms for specific target densities below one (see \cref{fig:density-edge-deletion}) or when the input is a clique or tree (see \cref{sec:poly-time-edge-del}).
However, beyond these cases, the problem turns out to be surprisingly hard.
There are target densities above or below one for which it is NP-hard, see \cref{fig:density-edge-deletion} for an overview.
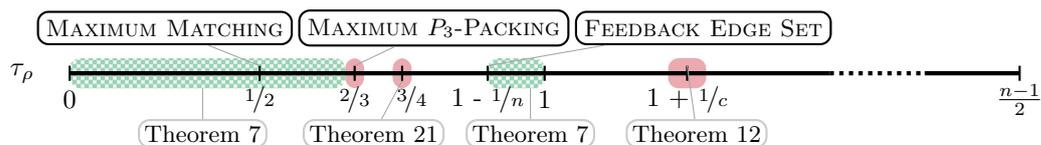
\begin{figure}
	\begin{tikzpicture}[xscale=1.25,node distance = 2pt]
		\def\probh{0.6}
		\def\halfPos{2}
		\def\twoThirdPos{3}
		\def\ThreeQuarterPos{3.5}
		\def\almostOnePos{4.4}
		\def\OnePos{5}
		\def\OnePlus{6.5}

		\tikzset{prob/.style={rectangle, rounded corners, draw=black, thick, inner sep=3pt}}
		\tikzset{thm/.style={rectangle, rounded corners, draw=black!20, thick, inner sep=2pt}}
		\tikzset{poly/.style={pattern={Hatch[angle=45,distance={3pt/sqrt(2)},xshift=.1pt,line width=1pt]}, pattern color=italyGreen!40,rounded corners}}
		\tikzset{nph/.style={fill=italyRed!40,rounded corners}}
		
		\fill[poly] (0,0.2) rectangle (\twoThirdPos - 0.1,-0.2);

		\fill[nph] (\twoThirdPos - 0.1,0.2) rectangle (\twoThirdPos + 0.1,-0.2);
		\fill[nph] (\ThreeQuarterPos - 0.1,0.2) rectangle (\ThreeQuarterPos + 0.1,-0.2);

		\fill[poly] (\almostOnePos,0.2) rectangle (\OnePos,-0.2);

		\fill[nph] (\OnePlus - 0.2,0.2) rectangle (\OnePlus + 0.2,-0.2);

		\node at  (-0.5,0)   {$\tau_\rho$};
		
		\draw[ultra thick] (0,0) -- (8,0);
		\draw[ultra thick,dotted] (8,0) -- (9,0);
		\draw[ultra thick] (9,0) -- (10,0);

		\draw[thick] (0,-0.1) -- (0,0.1);
		\node at  (0,-0.35)   {0};

		\node[] (triv1) at  (-0.5, \probh) {}; %

		\draw[thick] (\halfPos,-0.1) -- (\halfPos,0.1);
		\node at  (\halfPos,-0.35)   {\nicefrac{1}{2}};

		\node[prob,right=of triv1] (MM) {\small \textsc{Maximum Matching}};
		\draw[black!35] (\halfPos,0.1) -- (MM);

		\draw[thick] (\twoThirdPos,-0.1) -- (\twoThirdPos,0.1);
		\node at  (\twoThirdPos,-0.35)   {\nicefrac{2}{3}};

		\node[thm] at (\halfPos - 0.6, -0.8) (poly) {\small \cref{thm:edge-del-poly-density-intervals}};
		\draw[black!35] (\halfPos - 0.6,-0.2) -- (poly);

		\node[thm] at (\almostOnePos + 0.45, -0.8) (poly) {\small \cref{thm:edge-del-poly-density-intervals}};
		\draw[black!35] (\almostOnePos + 0.45,-0.2) -- (poly);

		\node[thm] at (\ThreeQuarterPos - 0.3, -0.8) (split) {\small \cref{thm:edgeDeletion-NP-hard-on-Split}};
		\draw[black!35] (\ThreeQuarterPos,-0.1) -- (split);

		\draw[thick] (\ThreeQuarterPos,-0.1) -- (\ThreeQuarterPos,0.1);
		\node at  (\ThreeQuarterPos + 0.1,-0.35)   {\nicefrac{3}{4}};

		\node[prob,right=of MM] (P3) {\small \textsc{Maximum $P_3$-Packing}};
		\draw[black!35] (\twoThirdPos,0.1) -- (P3);

		\draw[thick] (\almostOnePos,-0.1) -- (\almostOnePos,0.1);
		\node at (\almostOnePos,-0.35)   {1 - \nicefrac{1}{$n$}};

		\node[prob,minimum width=2.5cm,right=of P3] (FES) {\small \textsc{Feedback Edge Set}};
		\draw[black!35] (\almostOnePos,0.1) -- (FES);

		\draw[thick] (\OnePos,-0.1) -- (\OnePos,0.1);
		\node at (\OnePos,-0.35)   {1};

		\draw[thick] (\OnePlus,-0.1) -- (\OnePlus,0.1);
		\node at (\OnePlus,-0.35)   {1 + \nicefrac{1}{$c$}};

		\node[thm] (NP-h) at (\OnePlus + 0.1,-0.8) {\small \cref{thm:edgeDel-NP-Hard-rho>1}};
		\draw[black!35] (\OnePlus,0.1) -- (NP-h);
		
		\draw[thick] (10,-0.1) -- (10,0.1);
		\node at  (10,-0.35)   {$\frac{n-1}{2}$};

	\end{tikzpicture}
	\caption{
		The computational complexity and special cases of \densestEdgeDeletion{} for specific values of the target density~$\tau_\rho$, see \cref{sec:poly-time-edge-del,sec:np-hardness-edge-del} for the details.
		Green (hatched) boxes indicate polynomial-time solvable cases while red (solid) boxes denote NP-hard cases.
		The~$c$ in $1+1/c$ can be any constant larger than 24.
		The complexity for larger values of~$\tau_\rho$ remains open.
	}

 \label{fig:density-edge-deletion}
\end{figure}
We show that \densestEdgeDeletion{} remains NP-hard on claw-free cubic planar, planar bipartite, and split graphs (see \cref{sec:np-hardness-edge-del}).
Moreover, we prove W[1]-hardness with respect to the combined parameter~$k$ and feedback edge number.
This implies that the problem remains hard even on very sparse graphs as the feedback edge number in a connected graph is~$m - n + 1$, despite being polynomial-time solvable on trees.
Note that this also implies W[1]-hardness with respect to prominent parameters like treewidth.
Moreover, our employed reduction shows W[1]-hardness for $T_{h+1}$-\textsc{Free Edge Deletion}\footnote{Given a graph and an integer~$k$, the question is whether~$k$ edges can be removed so that no connected component has more than~$h$ vertices?} with respect to the treewidth, thus answering an open question by Enright and Meeks~\cite{EM2018}.
On the positive side, using integer linear programming, we classify the problem as fixed-parameter tractable with respect to the vertex cover number (see \cref{sec:edge-del-parameterized}).

Turning to \densestVertexDeletion{}, we derive NP-hardness for all~$\tau_\rho \in [0,n^{1-1/c}]$ for any constant~$c$.
Note that the density of a graph is between~$0$ and~$(n-1)/2$ and the case~$\tau_\rho = (n-1)/2$ is trivial.
Moreover, we show NP-hardness on planar bipartite graphs of maximum degree four, line graphs of planar bipartite graphs, and split graphs (see \cref{sec:vertex-del-NP-hardness}) as well as a polynomial-time algorithm for trees (see \cref{sec:vertex-del-poly-time}).
Furthermore, we prove W[2]-hardness with respect to~$k$ and fixed-parameter tractability with respect to the vertex cover number (see \cref{sec:vertex-del-parameterized}).
Notably, the latter algorithm is easier than in the edge deletion setting; in particular it does not rely on integer linear programming.

\subparagraph*{Further related work.}

The density as defined above is related to a variety of useful concepts. 
It belongs to a family of functions of the form $f(G,a,b,c)=a|E(G)|/(b|V(G)|-c)$ with $a,b,c\in \mathbb{Q}$. 
Depending on the values of~$a$, $b$ and $c$, the function has been used to study a variety of network properties~\cite{KHLL2009}. 
For instance $\rho(G)=f(G,1,1,0)$ is used in the study of random graphs~\cite{B1998}, whilst $f(G,1,1,1)$ comes up in the study of vulnerability of networks~\cite{C1985}, and $f(G,1,3,6)$ is used to study rigid frameworks~\cite{L1970}. 
A more general class of functions can be studied, where $a, b, c$ are not rational numbers but functions. 
For instance, Hobbs~\cite{H1991} studies the vulnerability of a graph $G$ by finding the subgraph $H$ of $G$ that maximizes $|E(H)|/(|V(H)|-\omega(H))$ with $\omega(H)$ being the number of connected components in $H$. The interpretation being that attacking such a subgraph would lead to the maximum number of connected components being created per unit of effort spent on the attack, where the effort is proportional to the number of edges being targeted by the attacker.

The maximum average degree $\mad(G)$ of a graph $G$ is the maximum of the average degrees of all subgraphs of~$G$. 
Note that~$\mad(G) = 2\rho^*(G)$. %
Recently, Nadara and Smulewicz~\cite{NS2022} showed that for every graph $G$ and positive integer $k$ such that $\mad(G) > k$, there exists a polynomial-time algorithm to compute a subset of vertices $S\subseteq V(G)$ such that $\mad(G-S) \leq \mad(G)- k$ and every subgraph of $G[S]$ has minimum degree at most $k-1$. 
Though no guarantees are given that $S$ has minimum size for subsets $S$ that achieve $\mad(G-S) \leq \mad(G)-k$.

Modifying a graph to bound its maximum average degree can be of use because of a variety of results on the colorability of graphs with bounded $\mad$. 
In general, $\mad$ can be used to give a bound on the chromatic number $\chi(G)$ of $G$, as $\chi(G)\leq \lfloor \mad(G)\rfloor+1$ \cite{NM2012}. 
It is well-known that for any planar graph $G$, the maximum average degree is related to the girth $g(G)$ of $G$ in the following way: $(\mad(G)-2)(g(G)-2)<4$~\cite{NS2022}.
Several results are known for variations of coloring problems and~$\mad$~\cite{BLP2014,BKY2013,BK2011,KY2017}.

\section{Preliminaries} 
\subparagraph*{Notation.}\label{sec:Notation}
For~$n \in \NN$ we set~$[n] = \{1,2,\ldots,n\}$. 
Let $G$ be a simple, undirected, and unweighted graph. 
We denote the set of vertices of $G$ by $V(G)$ and the set of edges of $G$ by $E(G)$. 
We set~$n_G = |V(G)|$ and $m_G = |E(G)|$. 
We denote the degree of a vertex~$v\in V(G)$ by~$\deg_G(v)$.
If the graph is clear from context, then we drop the subscript.
The minimum degree of $G$ is denoted by~$\delta(G)$, and the maximum degree of $G$ is denoted by~$\Delta(G)$.
We denote with~$H \subseteq G$ that~$H$ is a subgraph of~$G$.
The \emph{density of $G$} is~$\rho(G)=m/n$. 
We define the density of the empty graph as zero.
We denote by~$\rho^*(G)$ the density of the densest subgraph of $G$, that is, $\rho^*(G)=\max_{H \subseteq G}\rho(H)$. 
For a subset of vertices $W\subseteq V(G)$, we denote with $G[W]$ the subgraph \emph{induced} by $W$. %
For two subsets of vertices $W,U\subseteq V(G)$ we set $E(W,U)$ to be the set of edges with one endpoint in $W$ and another in $U$.

We denote by $P_n$ the path on $n$ vertices, by $K_n$ the complete graph on $n$ vertices (also called a clique of size $n$), and by $K_{a,b}$ the complete bipartite graph with $a$ and $b$ the size of its two vertex sets.
A graph $G$ is \emph{$r$-regular} if $\deg(v)=r$ for every vertex $v\in G$.
A \emph{perfect $P_3$-packing of G} is a partition of $V(G)$ into sets $V_1,V_2,\ldots,V_{n/3}$ such that for all $i\in [n/3]$ the graph $G[V_i]$ is isomorphic to $P_3$. 

A graph $G$ is \emph{balanced} if $\rho(G')\leq \rho(G)$ for every subgraph $G'\subseteq G$. %
Let $\rho'(G)=\frac{m}{n-1}$ for $1 \leq n-1$ and define $\rho'$ to be 0 for the empty graph and one-vertex graph. 
A graph $G$ is \emph{strongly balanced} if $\rho'(G')\leq \rho'(G)$ for every subgraph $G'\subseteq G$. 
Ruciński and Vince~\cite[page 252]{RV1986} point out that every strongly balanced graph is also balanced, though the converse is not true.

\subparagraph*{Problem Definitions.}
The problem definition for the edge deletion variant is as follows (the definition for vertex deletion is analogous):

\begin{problem}
    \problemtitle{\densestEdgeDeletion}
    \probleminput{A graph $G$, an integer $k\geq 0$ and a rational number $\tau_\rho \ge 0$.}
    \problemquestion{Is there a subset $F\subseteq E(G)$ with $|F|\leq k$ such that~$\rho^*(G- F)\leq\tau_\rho$?}
\end{problem}

There are two natural optimization problems associated to \densestEdgeDeletion{} which we call \densestEdgeDeletionMinRho{} (given~$k$ minimize~$\tau_\rho$) and \densestEdgeDeletionMinK{} (given~$\tau_\rho$ minimize the number of edge deletions~$k$). 

We emphasize that all problems for vertex deletion are defined and named analogously.

\subparagraph*{Useful Observations.}
We often compare the ratio of vertices to edges in different induced subgraphs. 
To this end, the following basic result is useful.
\begin{lemma}
	\label{lem:fractions} 
	$\frac{a}{b}\leq \frac{a+c}{b+d} \iff  \frac{a}{b}\leq \frac{c}{d}$  and $\frac{a}{b} = \frac{a+c}{b+d} \iff \frac{a}{b} = \frac{c}{d}$.
\end{lemma}
\appendixproof{lem:fractions}
{
\begin{proof}
	For the inequality:
	$$\frac{a}{b} \leq \frac{a+c}{b+d} \iff a(b+d) \leq b(a+c) \iff ad \leq cb \iff \frac{a}{b} \leq \frac{c}{d}$$
	The proof for equality is analogous.
\end{proof}
}
The following is a collection of easy observations that can be obtained with \cref{lem:fractions}.

\begin{lemma}
	\label{lem:nice-props} 
	Let~$G^*$ be a densest subgraph of~$G$ with~$\rho(G^*) = \rho^*(G)$.
	Then:
	\begin{enumerate}
		\item If~$G^*$ is not connected, then each connected component~$C$ of~$G^*$ has density~$\rho^*(G)$. \label{niceprop-densest-subgraph-connected}
		\item If~$\rho(G^*) = a / b$ for $a,b \in \NN$ and~$a < b$, then~$a = b-1$ and~$G^*$ is a tree on~$b$ vertices or a forest where each tree is on~$b$ vertices. \label{niceprop-trees-density}
		\item Any vertex~$v \notin V(G^*)$ has at most~$\lfloor\rho(G^*)\rfloor$ neighbors in~$V(G^*)$. \label{niceprop-maxneighbors-in-densest}
		\item The minimum degree in~$G^*$ is at least $\lceil\rho(G^*)\rceil$. This is tight for trees. \label{niceprop-minneighbors-within-densest}
	\end{enumerate}
\end{lemma}
\appendixproof{lem:nice-props}
{
\begin{proof}
	\begin{enumerate}
		\item Let~$C$ be a connected component of~$G^*$ with~$a$ edges and~$b$ vertices and let~$c = m_{G^*}-a$ and~$d=n_{G^*}-b$. %
			If~$a/b > \rho(G^*)$, then~$G^*$ was not the densest subgraph, a contradiction.
			If~$a/b < \rho(G^*) = (a+c)/(b+d)$, then~$G[V(G^*) \setminus C]$ has, by \cref{lem:fractions}, density~$c/d > \rho(G^*)$.
			Again, $G^*$ was not the densest subgraph, a contradiction.
			Thus, $a/b = \rho(G^*)$.
		\item 
			If $\rho(G^*) = a / b$ for $a,b \in \NN$ and~$a < b$, then $G^*$ cannot contain a cycle, as any cycle has density $1$. 
			Therefore, $G^*$ is a forest, and so is every subgraph of $G$. 
			A tree on $n$ vertices has density $(n-1)/n$. 
			Since $f(n)=(n-1)/n$ is an increasing function, the largest tree in $G^*$ has density exactly $a/b$. 
			Therefore $a=b-1$. 
			Moreover, by \cref{lem:nice-props} (\ref{niceprop-densest-subgraph-connected}.) no other tree in $G^*$ can have less than $a$ edges and $b$ vertices. 
        \item 
			Suppose the statement is false.
			Let~$\rho(G^*) = a/b$ for some~$a,b \in \NN$.
			Then, by \cref{lem:fractions},
			\[\rho(G^*) = \frac{a}{b} < \frac{a + |E(\{v\},V(G^*))|}{b + 1} = \rho(G^*\cup\{v\}).\]
			Thus, $G^*$ was not the densest subgraph, a contradiction.
        \item 
			If there exists a vertex $v\in V(G^*)$ with $\deg_{G^*}(v)<\rho(G^*)$ then, similar to the previous case, $\rho(G^*\backslash\{v\})>\rho(G^*)$ which is a contradiction. Therefore, $\delta(G^*)\geq \rho(G^*)$ and since the minimum degree is an integer, $\delta(G^*)\geq \lceil\rho(G^*)\rceil$. This is tight for trees. 
    \end{enumerate} 
\end{proof}
}

Note that \cref{lem:nice-props} (\ref{niceprop-minneighbors-within-densest}.) implies that we can remove vertices with degree less than our desired target density~$\tau_\rho$.

\begin{rrule}\label{rr:low-degree}
	Let~$v$ be a vertex with degree~$\deg(v) < \tau_\rho$. Then delete~$v$.
\end{rrule}

The following two observations imply that there are only a polynomial number of ``interesting'' values for the target density~$\tau_\rho$.
Thus, if we have an algorithm for the decision problem, then, using binary search, one can solve the optimization problems with little overhead in running time.

\begin{observation}\label{obs:possible_density_on_n_vertices}
	The density of a graph $G$ on $n$ vertices can have values between $0$ and $(n-1)/2$ in intervals of $1/n$:
	$\rho(G) \in \{ 0, 1/n, 2/n, \ldots, \binom{n}{2}/n = (n-1)/2 \}$.
\end{observation}

\begin{observation}[\cite{G1984}]\label{obs:min_distance_between_densities}
    The maximum density of a subgraph of $G$ can take only a finite number of values: $\rho^*(G)\in \{m' / n'\mid 0\leq m'\leq m, 1 \leq n'\leq n\}$. 
    Moreover, the minimum distance between two different possible values of $\rho^*(G)$ is at least $1/(n(n-1))$.
\end{observation}

\appendixsection{sec:Notation}

\toappendix{

\begin{proposition}\label{prop:equiv} 
\densestEdgeDeletionMinRho{} is polynomial-time solvable if and only if  \densestEdgeDeletionMinK{}  is polynomial-time solvable. The same equivalence holds between \densestVertexDeletionMinRho{} and  \densestVertexDeletionMinK{}.
\end{proposition}

\begin{proof}
To see this, note that if \densestEdgeDeletionMinK{} is polynomial time solvable then we can solve an instance $\mathcal{I}=(G,k)$ of \densestEdgeDeletionMinRho{} by running the algorithm for \densestEdgeDeletionMinK{} for instances $(G,\rho_i)$ using binary search with $\rho_i\in  \{m'/n'\mid 0\leq m'\leq m_G, 1 \leq n'\leq n_G\}$   (set of values of from \cref{obs:min_distance_between_densities}) %
then returning the smallest $\rho_i$ for which the algorithm for \densestEdgeDeletionMinK{} returned a subset of size at most $k$. The  other direction is analogous: we need to solve $\log(m_G + 1)$ instances of \densestEdgeDeletionMinRho{}, to do binary search on $k\in \{0,1,\ldots,m_G\}$, and return the smallest $k$ for which \densestEdgeDeletionMinRho{} returned a density of at most $\tau_\rho$.  
\end{proof}
}

\section{\densestEdgeDeletion{}}\label{Sec:Critical_k_edge_deletion}

In this section we provide our results for \densestEdgeDeletion{}, starting with the polynomial-time algorithms, continuing with the NP-hard cases, and finishing with our parameterized results.

\subsection{Polynomial-time solvable cases}\label{sec:poly-time-edge-del}

\subparagraph*{Specific Density Intervals.}

We show that if the density~$\tau_\rho$ falls within one of two intervals, then \densestEdgeDeletionMinK{} boils down to computing maximum matchings or spanning trees.

\begin{theorem}\label{thm:edge-del-poly-density-intervals}
	\densestEdgeDeletionMinK{} can be solved in time $O(m\sqrt{n})$ if~$0 \le \tau_\rho < 2/3$ and in time~$O(n+m)$ if $1-1/n\leq \tau_\rho\le1$.
\end{theorem}

\begin{proof}
	The proof is by case distinction on~$\tau_\rho$.
	
	Note that if~$\tau_\rho < 1/2$, then no edge can remain in the graph as a~$K_2$ has density~$1/2$.
	Similarly, if $1/2 \le \tau_\rho < 2/3$, then no connected component can have more than one edge: Otherwise, the component would contain a~$P_3$ which has density~$2/3$.
	Hence, computing a maximum cardinality matching in time $O(m\sqrt{n})$~\cite{MV1980} and removing all edges not in the matching solves the given instance of \densestEdgeDeletionMinK{}.
	
	The second interval is similar.
	If~$1 - 1/n \le \tau_\rho < 1$, then the resulting graph cannot have any cycle as a cycle has density~$1$. 
	Moreover, any tree on at most~$n$ vertices has density at most~$1-1/n$.
	Thus, in this case \densestEdgeDeletionMinK{} is equivalent to computing a minimum feedback edge set, which can be done in time $O(n+m)$ by e.\,g.\ deleting all edges not in a spanning tree.
	
	Lastly, if~$\tau_\rho = 1$, then each connected component can have at most one cycle, that is, the resulting graph must be a pseudoforest: 
	Consider a connected component~$C$ with~$\ell$ vertices and at least two cycles.
	Any spanning tree of~$C$ contains~$\ell-1$ edges and misses at least one edge per cycle.
	Hence, $C$ contains at least~$\ell+1$ edges and has, thus, density larger than one.
	Thus, each connected component in the remaining graph can have at most as many edges as vertices.
	Hence, a solution to the \densestEdgeDeletionMinK{} instance is to do the following for each connected component: delete all edges not in a spanning tree and reinsert an arbitrary edge.
	This can be done in $O(n+m)$ time.
\end{proof}

\subparagraph*{Trees.}

If the input is a tree, then any target threshold~$\tau_\rho \ge 1$ makes the problem trivial.
Hence, the case~$\tau_\rho < 1$ is left.
Thus, each tree in the remaining graph can have at most~$h = \lfloor 1 / (1 - \tau_\rho)\rfloor$ many vertices: 
a tree with~$h' > h$ vertices has density~$(h'-1)/h' = 1 - 1/h' > 1 - 1/h \ge 1 - (1-\tau_\rho) = \tau_\rho$.
Hence, the task is to remove as few edges as possible so that each connected component in the remaining graph is of order at most~$h$.
This problem is known as $T_{h+1}$-\textsc{Free Edge Deletion} and can be solved in~$O((wh)^{2w} n)$ time~\cite{EM2018}, where~$w$ is the treewidth.
As trees have treewidth one and~$h \le n$ (otherwise the problem is trivial), we get the following.

\begin{theorem}\label{thm:EdgeDel-poly-on-trees}
	On the trees, \densestEdgeDeletionMinK{} can be solved in time $O(n^3)$.  
\end{theorem}

\subparagraph*{Cliques.}
The problem is not completely trivial on cliques: 
While a target threshold~$\tau_\rho$ indicates an upper bound on the remaining edges (as~$\rho(G-F) \le \tau_\rho$ must hold), the question is whether for all~$m$ and~$n$ there is a balanced graph~$G$ with~$m$ edges and~$n$ vertices; recall that a graph is balanced if the whole graph is a densest subgraph.
Ruciński and Vince~\cite{RV1986} showed a slightly stronger statement about strongly balanced graphs. 
Recall that every strongly balanced graph is also balanced; refer to \cref{sec:Notation} for formal definitions.

\begin{theorem}[{\cite[Theorem~1]{RV1986}}]\label{thm:strongly_balanced}
	Let $n$ and $m$ be two integers. If $1 \leq n-1 \leq m \leq  \binom{n}{2}$, then there exists a strongly balanced graph with $n$ vertices and $m$ edges.
\end{theorem}

The proof of Ruciński and Vince~\cite{RV1986} is constructive:
A strongly balanced graph (that is also a balanced graph) with $m$ edges and $n$ vertices can be constructed in $O(m)$ time.  

\begin{theorem}
	\label{thm:EdgeDel-poly-on-cliques}
	On the complete graph $K_n$, \densestEdgeDeletionMinK{} and \densestEdgeDeletionMinRho{} can be solved in time $O(n^2)$.  
\end{theorem}
{
\begin{proof}	
	We provide the proof for \densestEdgeDeletionMinK{}. The proof for \densestEdgeDeletionMinRho{} is analogous.

	Consider an instance $(G = K_n,\tau_\rho)$ of \densestEdgeDeletionMinK{}. 
	Let $F \subseteq E(G)$ be a solution that our algorithm wants to find. 
	Throughout the proof we assume $\tau_\rho \leq (n-1)/2$ and~$n\ge 1$, as otherwise $F=\emptyset$.
	We consider two cases: $\tau_\rho<1$ and~$\tau_\rho \ge 1$.

	\emph{Case 1. ($\tau_\rho<1$):} 
	Let $t \leq n$ be the largest integer satisfying $(t-1)/t\leq  \tau_\rho$.
	By \cref{lem:nice-props} (\ref{niceprop-trees-density}.), the resulting graph~$K_n-F$ must be a collection of trees on at most~$t$ vertices each. 
	Thus, partition the vertices in~$\lceil n/t \rceil$ parts of size at most~$t$.
	For each part keep an arbitrary spanning tree.
	Then~$F$ consists of all non-kept edges, i.\,e., all edges between the parts and all edges not in the selected spanning trees.
	Clearly, this can be done in~$O(n^2)$ time.

	\emph{Case 2. ($\tau_\rho \ge 1$):}
	We will construct a strongly balanced graph $G'$ on $n$ vertices with density as close to $\tau_\rho$ as possible, as allowed by \cref{obs:possible_density_on_n_vertices}.
	To this end, let $t$ be the largest number in~$\{0,1/n,\ldots, (n-1)/2\}$ so that $t\leq \tau_\rho$, thus $t=\ell/n$ for some integer $\ell \in \{n,n+1,\ldots,\binom{n}{2}\}$ (as~$\tau_\rho \ge 1$).
	
	Now we use \cref{thm:strongly_balanced} to construct a balanced graph $G'$ with $\ell$ edges and $n$ vertices, thus $\rho^*(G') \le \tau_\rho$.
	We will select~$F \subseteq E(G)$ so that~$G'=G-F$, that is, $F$ contains all edges not in~$G'$ and~$|F| = \binom{n}{2}-\ell$.
	By choice of $\ell$ we know that $(\ell+1)/n>\tau_\rho$. 
	Hence, removing less than $\binom{n}{2}-\ell$ edges from $G$ means the whole graph has density more than~$\tau_\rho$. 
	As building the graph $G'$ takes time $O(\ell)$ the overall running time is $O(n^2)$.

	To construct the analogous proof for \densestEdgeDeletionMinRho{} use two cases: $k> \binom{n}{2} - n$ (equivalent of Case 1 above) and $k \le \binom{n}{2} -n$ (equivalent of Case 2). 
\end{proof}
}

\subsection{NP-Hardness for special graph classes}\label{sec:np-hardness-edge-del}

\subparagraph*{Claw-free cubic planar graphs.}
For our first hardness proof of \densestEdgeDeletion{} we provide a reduction from \textsc{Perfect $P_3$-packing} which was proven \NP-complete even in claw-free cubic planar graphs~\cite{XL2021}. 
Denote with $P_k$ a path on $k$ vertices. 
A \emph{perfect $P_3$-packing} of a given graph $G$ is a partition of $G$ into subgraphs in which each subgraph is isomorphic to $P_3$. 
The \textsc{Perfect $P_3$-packing} problem is defined as follows:

    \begin{problem}\label{prob:perfect_P3_packing}
        \problemtitle{\textsc{Perfect $P_3$-packing}}
        \probleminput{A graph $G$.}
        \problemquestion{Is there a perfect $P_3$-packing of $G$?}
    \end{problem}

\begin{theorem}\label{thm:densestEdgeDel-is-NP-complete}
\densestEdgeDeletion{} is \NP-complete for $\tau_\rho = 2/3$ even on claw-free cubic planar graphs. 
\end{theorem}

\begin{proof}
    Given an instance $G$ of \textsc{perfect $P_3$-packing} where~$G$ is a claw-free cubic planar graph, let $\mathcal{I} = (G,k,\tau_\rho)$ be an instance of \densestEdgeDeletion{} where $k = m - 2n/3$ and $\tau_\rho = 2/3$.
    We claim that~$G$ has a perfect $P_3$-packing if and only if there is a set~$F \subseteq E(G)$ with $\rho^*(G-F)=2/3$ and $|F|=m - 2n/3$. 
    
    ``$\Rightarrow:$'' Given a perfect $P_3$-packing of $G$, we set $F$ to be the set of all edges in $G$ that are not in the $P_3$-packing. Clearly $\rho^*(G-F)=2/3$. Additionally, $|F|=m - 2n/3$ since there are $n/3$ paths in a perfect $P_3$-packing, and each path has two edges. 

    ``$\Leftarrow:$'' 
    Consider $F\subseteq E(G)$ of size at most $m - 2n/3$ such that $\rho^*(G-F)\leq 2/3$. 
    Then $|E(G)  \setminus F| \geq 2n/3$. 
    Note that $\tau_\rho=2/3$ implies, by \cref{lem:nice-props} (\ref{niceprop-trees-density}.), that all connected components of $G-F$ must be trees of size at most $3$. 
    In other words, all connected components in $G-F$ are singletons, $P_2$'s, or $P_3$'s. 
    Denote by $t$ the number of connected components that are $P_3$'s in $G-F$. 
    Then there are $3t$ vertices and $2t$ edges of $G-F$ which belong to a $P_3$. 
    Consequently, there are $n-3t$ vertices and $|E(G) \setminus F|-2t\geq 2(n-3t)/3$ edges of $G-F$ which are either singletons or belong to a $P_2$. 
    This is possible only if $n-3t = 0$, that is $G-F$ is a perfect $P_3$-packing of $G$.
\end{proof}

\subparagraph*{Planar graphs with $\Delta=3$, target density above 1.} 
We next show that \densestEdgeDeletion{} remains \NP-complete even for $\tau_\rho > 1$. 
To prove this, we reduce from \textsc{Vertex Cover} on cubic planar graphs, which is known to be \NP-complete~\cite{garey_et_al:1974:VC_hard_on_cubic}.  
\textsc{Vertex Cover} is defined as follows:
	\begin{problem}
        \problemtitle{\textsc{Vertex Cover}}
        \probleminput{A graph $G$ and a positive integer $k$.}
        \problemquestion{Is there a vertex cover $C\subseteq V(G)$ of size at most $k$, that is, for each $\{u,v\} \in E$ at least one of $u$ or $v$ is in~$C$?}
    \end{problem}

\begin{theorem}
	\label{thm:edgeDel-NP-Hard-rho>1}
	\densestEdgeDeletion{} is \NP-complete for $\tau_\rho>1$ even on planar graphs with maximum degree $3$.
\end{theorem}
\appendixproof{thm:edgeDel-NP-Hard-rho>1}
{
\begin{proof}
	Given an instance $(G,k)$ of \textsc{Vertex Cover}, where $G$ is a planar cubic graph~\cite[p.\ 195]{GJ79}, we construct an instance $(G', k, 1 + 1/(6+\ell)-\varepsilon)$  of \densestEdgeDeletion{} where $\varepsilon>0$ (e.\,g. $\varepsilon = 1/(n_{G'} + 1)$), $\ell \geq 19$ is an integer, and $G'$ is constructed as follows: 
    For every vertex $v\in V(G)$ with neighbors $x,y,z$, we construct a triangle $T_v$ with vertex set $\{v_x,v_y,v_z\}$.
    For every edge $\{u,v\}\in E(G)$ we join $T_u$ and $T_v$  via a path on $\ell$ vertices as follows: we create $\ell$ new vertices $w_1,\ldots,w_\ell$ and construct the path $P_{uv} = (u_v, w_1, w_2, \ldots, w_\ell, v_u)$ where $u_v\in T_u$, and $v_u\in T_v$.
    This concludes the construction of $G'$.
 
    Note that $G'$ has a vertex set of size $n_{G'} = 3n_G + \ell m_G$, an edge set of size $m_{G'} = 3n_G + (\ell+1)m_G$, and maximum degree $3$. 
    Moreover, each edge $\{u,v\}\in E(G)$ corresponds to the induced subgraph $G'[T_u\cup T_v \cup P_{uv}]$ depicted in \cref{fig:hourglass}, and $\rho(G'[T_u\cup T_v \cup P_{uv}]) = 1+1/(6+\ell)$.
    
    \begin{figure}[t]
        \centering
            \begin{tikzpicture} 
                \tikzset{vertex_small/.style={circle,draw,fill, inner sep=2pt}}
                \pgfmathsetmacro{\dist}{2.5}
                \foreach \pos/\name in {
            		 (-2,-0.5)/x_a,   (-1,0)/x_v,   (-2,0.5)/x_b,  %
                  (5.75,0.5)/z_a,  (4.75,0)/z_v,   (5.75,-0.5)/z_b} {%
                    \node[vertex_small] (\name) at \pos {};
                  }
                \foreach \x in {1,2,...,19}%
                {\draw[fill] ({-0.6+\x/4},0) circle (0.07);    }
                  \foreach \u/\v in {
                                    {x_a/x_b}, {x_b/x_v}, {x_v/x_a},
                                    {z_a/z_b}, {z_b/z_v}, {z_v/z_a}}{
                        \draw[thick] (\u) -- (\v);
                        }
                \foreach \u/\v in {
                                    {x_v/z_v}}{
                        \draw[thick] (\u) -- (\v);
                        }
                \node[very thick, italyGreen, draw,rounded corners, dashed,fit=(x_a) (x_b)(x_v),label=above:$T_u$] {};
                \node[very thick, italyGreen, draw,rounded corners, dashed,fit=(z_a) (z_b)(z_v),label=above:$T_v$] {};
                \node[very thick, italyRed, draw,rounded corners, dashed,fit=(x_v)(z_v),label=above:$P_{uv}$] {};
                    
            \end{tikzpicture}
        \caption{
			Two $K_3$ joined by a path of exactly $\ell=19$ vertices of degree $2$. 
			This subgraph of $G'$ has a correspondance to an edge $\{u,v\}$ of the original graph $G$, where $T_u$ (in green) corresponds to $u$, $T_v$ (in green) corresponds to $v$ and $P_{uv}$ (in red) corresponds to edge $\{u,v\}$.
			Note that this graph has density~$(6+\ell + 1)/(6+\ell) = 1 + 1/(6+\ell) > \tau_\rho$.}
        \label{fig:hourglass}
    \end{figure}
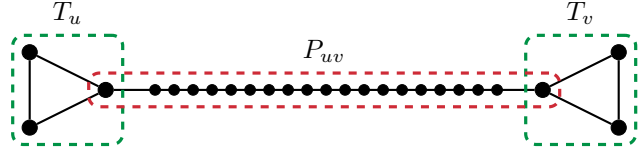

	We claim that~$(G,k)$ is a yes-instance of \textsc{Vertex Cover} if and only if $(G', k, 1+1/(6+\ell) - \varepsilon)$ is a yes-instance of \densestEdgeDeletion{}.

	``$\Leftarrow:$'' Consider an edge-set $F$ of size $k$, solution for the instance $(G', k, 1+1/(6+\ell) - \varepsilon)$ of \densestEdgeDeletion{}. Since each edge $\{u,v\}$ in $G$ corresponds to $G'[T_u\cup T_v \cup P_{uv}]$ of density $1+1/(6+\ell)$, at least one edge of $G'[T_u\cup T_v \cup P_{uv}]$ must be in $F$. We now construct a vertex cover $C$ of $G$ as follows: For each $e'\in F$ we add vertex $v$ to $C$ if $e'$ belongs to $T_v$; if $e'$ belongs to $P_{uv}$, arbitrarily put either $u$ or $v$ in $C$. This procedure ensures that for every edge $\{u,v\}$ in $G$ either $u$ or $v$ is in $C$. Thus $C$ is a vertex cover of size $|C|\leq |F| \leq k$.

	``$\Rightarrow:$'' Let $C$ be a vertex cover of size $k$ in $G$. We build an edge-deletion set $F$ for the instance $(G', k, 1+1/(6+\ell) - \varepsilon)$ as follows: 
    For each $v\in C$ add an arbitrary edge $e'$ from $T_v$ to $F$. 

    Let $H$ be a densest subgraph of $G'-F$. Assume $H$ is connected. Otherwise, treat connected components separately (see \cref{lem:nice-props} (\ref{niceprop-densest-subgraph-connected}.)). 
    We now prove by contradiction that $H$ has density $\rho(H)< 1+1/(6+\ell)$.

	Suppose for the sake of contradiction that $H$ has density $\rho(H)\geq 1+1/(6+\ell)>\tau_\rho$. 
	By \cref{lem:nice-props} (\ref{niceprop-minneighbors-within-densest}.), no vertex of degree $1$ belongs to $H$. 
	Thus, the vertex set $V(H)$ consists of the vertices from some triangles~$T_{v_1},\ldots,T_{v_x}$ and some paths connecting these triangles and~$F$ contains edges from some of the triangles, see \Cref{fig:Q_region_RHO_g1} for an illustration.
	\begin{figure}[t]
		\centering
  		\begin{tikzpicture}
            \tikzset{vertex_small/.style={circle,draw, fill, inner sep=2pt}}
            \tikzset{vertex_large/.style={circle,draw=black,dashed}}
            \tikzset{quasi_edge/.style={draw=black,dotted, ultra thick}}
            \tikzset{region_Q/.style={black!25,line width=25pt,rounded corners=5pt}}
            \tikzset{region_S/.style={blue!40,line width=15pt,rounded corners=5pt}}

        		 \foreach \pos/\name in {
        		 {(0,0)/b-1}, {(1,0)/b-2}, {(0.5,1)/b-3},
        		 {(3,0)/b-4}, {(4,0)/b-5}, {(3.5,1)/b-6},
        		 {(6,0)/b-7}, {(7,0)/b-8}, {(6.5,1)/b-9},
        		 {(5.5,5)/b-11}, {(5,4)/b-12}, {(4.5,5)/b-13}}
        			   \node[vertex_small] (\name) at \pos {};
        		   \foreach \one/\two in {{b-1/b-2}, {b-2/b-3}, {b-4/b-5}, {b-5/b-6}, {b-7/b-8}, {b-8/b-9},{b-11/b-12}, {b-12/b-13}}
        			   \draw[thick] (\one) -- (\two);
        		 \foreach \pos/\name  in {
        		 {(-3,0)/t-1}, {(-2,0)/t-2}, {(-2.5,1)/t-3},
        		 {(0.5,2)/t-4}, {(0,3)/t-5}, {(1,3)/t-6},
        		 {(5,3)/t-7}, {(4.5,2)/t-8}, {(5.5,2)/t-9},
        		 {(9,0)/t-10}, {(10,0)/t-11}, {(9.5,1)/t-12}}
        			   \node[vertex_small] (\name) at \pos {};
        		   \foreach \one/\two in {{t-1/t-2}, {t-2/t-3},{t-3/t-1}, {t-4/t-5}, {t-5/t-6},{t-6/t-4}, {t-7/t-8}, {t-8/t-9}, {t-9/t-7},{t-10/t-11},{t-11/t-12}, {t-12/t-10}}
        			   \draw[thick,blue] (\one) -- (\two);

        		  \draw[quasi_edge] (b-2)--(b-4); 
        		  \draw[quasi_edge] (b-5)--(b-7); 
        		  
        		  \draw[quasi_edge] (t-2)--(b-1);
        		  \draw[quasi_edge] (t-4)--(b-3); 
        		  \draw[quasi_edge] (t-8)--(b-6); 
        		  \draw[quasi_edge] (t-9)--(b-9); 
        		  \draw[quasi_edge] (t-10)--(b-8);
        		  \draw[quasi_edge] (t-7)--(b-12);
                \draw[quasi_edge] (t-6)--(b-13);
                \draw[quasi_edge] (t-12)--(b-11);
        \end{tikzpicture}
		\caption{
			A schematic illustration of a densest subgraph $H$. The dotted edges correspond to paths of $\ell$ vertices. 
			The blue edges indicate triangles that correspond to an independent set in the original graph.
		}
		\label{fig:Q_region_RHO_g1}
    \end{figure}
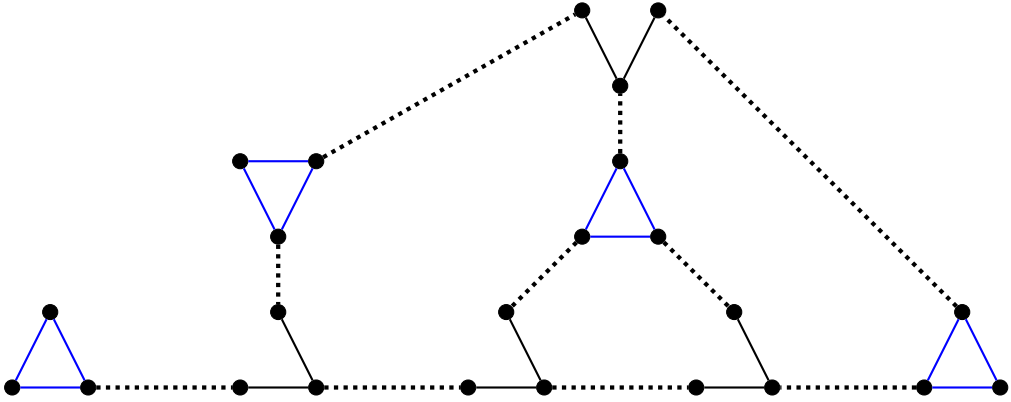
	Hence, $H$ corresponds to the subgraph~$G[V']$ with~$V' = \{v_1, \ldots, v_x\}$.
	Since~$H$ is connected it follows that~$G[V']$ is connected.
	Let~$y$ be the number of edges in~$G[V']$ (equal to the number of connecting paths in~$H$).
	Note that since~$G$ has maximum degree three, we have that at least~$x/4$ vertices of~$V'$ in any vertex cover of~$G[V']$.
	Moreover, $y \le 3x/2$.
	From this we can derive the following bounds for the number of edges and vertices:
	\begin{align*}
		n_H & = 3x + y \ell  & m_H & \le 2\cdot(x/4) + 3\cdot(3x/4) + y(\ell + 1)
	\end{align*}
	Thus, 
	\begin{align*}
		\rho(H) \le \frac{11x/4 + y(\ell + 1)}{3x + y \ell} = 1 + \frac{y - x/4}{3x + y\ell} 
	\end{align*}
	To obtain the contradiction it remains to show that:
	\begin{align*}
		&& \frac{y - x/4}{3x + y\ell} & < \frac{1}{6+\ell}  &\iff && 6y & < \frac{9x}{2} + \frac{x \ell}{4}
	\end{align*}
	Since~$y \le 3x/2$ and assuming~$\ell > 18$, we get:
	\begin{align*}
		6y \le 9x = \frac{9x}{2} + \frac{18 x}{4} < \frac{9x}{2} + \frac{x \ell}{4}
	\end{align*}
\end{proof}

\begin{corollary}\label{cor:NP-completeness_G-F_bigger_1}
	\densestEdgeDeletion{} is \NP-complete on planar graphs where $m - k > n$.
\end{corollary}
\begin{proof}
	Let $(G,k)$ be an instance of {\sc Vertex cover} with $G$ cubic, then $m= 3n/2$. Note that $k\leq n$. 
	The graph $G'$ from the proof of \cref{thm:edgeDel-NP-Hard-rho>1} has~$3n + m\ell$ vertices and~$3n + m(\ell +1)$ edges.
	Thus, for~$G'$ we have~$m_{G'} - n_{G'} = m = 3n/2 > k$.
\end{proof}
}

\subparagraph*{Bipartite graphs and split graphs.}\label{sec:Bipartite_EdgeDeletion}
Finally, we show that \densestEdgeDeletion{} is \NP-hard even on planar bipartite graphs and on split graphs.

\begin{theorem}
	\label{thm:hardness-densest-edge-bipartite}
    \densestEdgeDeletion{} remains \NP-complete on planar bipartite graphs.
\end{theorem}

\appendixproof{thm:hardness-densest-edge-bipartite}
{
Firstly, we introduce concepts that will be used to prove \cref{thm:hardness-densest-edge-bipartite}.
The \emph{subdivision} of an edge $\{u,v\}\in E(G)$ is the operation of creating new vertex $w$, adding edges $\{u,w\}$ and $\{v,w\}$ to $E(G)$ and deleting edge $\{u,v\}$. 
We denote by $S(G)$ the graph obtained by subdividing each edge of $G$ once.
The graph $S(G)$ is called the \emph{subdivision graph} of $G$ or the \emph{subdivision} of $G$.
The reverse operation of subdividing an edge, is \emph{smoothing} a vertex $w\in V(G)$ of degree $2$ with incident edges $\{u,w\},\{w,v\}$. The smoothing a vertex operation consists in adding edge $\{u,v\}$ to $E(G)$ and deleting vertex $w$.

The reduction is done from \densestEdgeDeletion{} using the fact that the subdivision $S(G)$ of any graph $G$ is bipartite,  and the fact, shown in this section, that the subdivision of a densest subgraph of $G$ is itself a densest subgraph of $S(G)$. 

By the definition of the subdivision graph $S(G)$, we get the following observation. 

\begin{observation}\label{obs:density_of_subdivision}
	Let $G$ be a graph and $S(G)$ its subdivision. Then, $S(G)$ has $n_{S(G)}=n_G+m_G$ vertices and $m_{S(G)}=2m_G$ edges.
\end{observation}

\begin{lemma}\label{lem:subdivision_of_subgraph}
	Let $H$ be a densest subgraph of $G$. Then, $\rho(S(H)) \geq \rho(S(G'))$ for any $G'\subseteq G$.
\end{lemma}

\begin{proof}
	Let $H$ and $G'$ be two subgraphs of $G$, with $H$ a densest subgraph of $G$.
    By \cref{obs:density_of_subdivision}, we know that $\rho(S(H))=2m_H/(n_H + m_H)$,  and $ \rho(S(G'))=2m_{G'}/(n_{G'}+m_{G'})$. Thus, after substituting and rearranging terms we get that:
	\begin{equation}\label{eq:subdivision_inequality_density_computation}
		\rho(S(H)) \geq \rho(S(G')) 
		\iff 
		\frac{m_H}{n_H}  \geq \frac{m_{G'}}{n_{G'}}
	\end{equation}
    The last inequality is true since $H$ is a densest subgraph of $G$. Thus any subgraph in $S(G)$ which is a subdivision of a subgraph of $G$ is at most as dense as $S(H)$.
\end{proof}

As illustrated in \cref{fig:densest_not_subdivision}, a densest subgraph of $S(G)$ does not need to be a subdivision of a subgraph of $G$. We show in \cref{lem:densest_subgraph_is_subdivision_rho_not_1}  that this can only happen if $\rho^*(S(H))=1$.
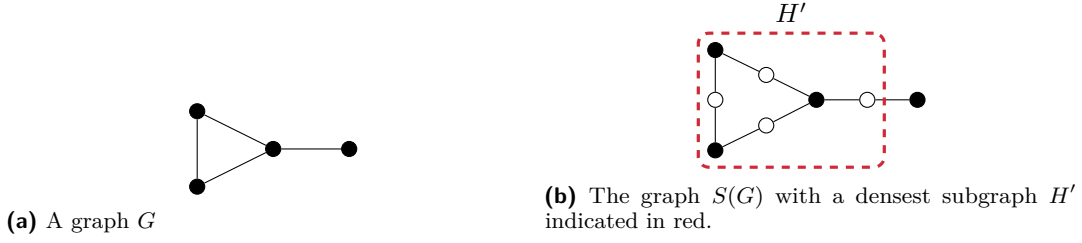
\begin{figure}[t]
    \begin{subfigure}[b]{0.50\textwidth}
        \centering
            \begin{tikzpicture}
              \node[draw, circle, fill, inner sep=2pt] at (0, 0) (t1) {};
              \node[draw, circle, fill, inner sep=2pt] at (0, 1) (t2) {};
              \node[draw, circle, fill, inner sep=2pt] at (1, 0.5) (t3) {};
              \node[draw, circle, fill, inner sep=2pt] at (2, 0.5) (t4) {};

              \draw[-] (t1) -- (t2);
              \draw[-] (t2) -- (t3);
              \draw[-] (t3) -- (t1);
              \draw[-] (t3) -- (t4);
            \end{tikzpicture}
        \caption{A graph $G$}
    \end{subfigure}
    \begin{subfigure}[b]{0.50\textwidth}
        \centering
        \begin{tikzpicture}
          \node[draw, circle, fill, inner sep=2pt] at (0, 0) (t1) {};
          \node[draw, circle, inner sep=2pt] at (0, 0.67) (b1) {};
          \node[draw, circle, fill, inner sep=2pt] at (0, 1.33) (t2) {};
          \node[draw, circle, inner sep=2pt] at (0.67, 1) (b2) {};
          \node[draw, circle, fill, inner sep=2pt] at (1.33, 0.67) (t3) {};
          \node[draw, circle, inner sep=2pt] at (2, 0.67) (b3) {};
          \node[draw, circle, fill, inner sep=2pt] at (2.66, 0.67) (t4) {};
           \node[draw, circle, inner sep=2pt] at (0.67, 0.33) (b4) {};

          \draw[-] (t1) -- (b1);
          \draw[-] (b1) -- (t2);
          \draw[-] (t2) -- (b2);
          \draw[-] (b2) -- (t3);
          \draw[-] (t3) -- (b4);
          \draw[-] (b4) -- (t1);
          \draw[-] (t3) -- (b3);
          \draw[-] (b3) -- (t4);

          \node[very thick, italyRed, draw,rounded corners, dashed,fit=(t1) (t2) (t3) (b3),label=above:$H'$] {};
        \end{tikzpicture}
        \caption{The graph $S(G)$ with a densest subgraph $H'$ indicated in red.}
    \end{subfigure}
    \caption{An example of a graph $G$ and a densest subgraph $H'$ in $S(G)$ which is not a subdivision of a subgraph of $G$.}
    \label{fig:densest_not_subdivision}
\end{figure}

\begin{lemma}\label{lem:densest_subgraph_is_subdivision_rho_not_1}
	A densest subgraph $H'$ of $S(G)$ with density $\rho(H')\neq 1$ is a subdivision of a subgraph of $G$.
\end{lemma}
\begin{proof}
	The proposition is trivially true when $G$ has no edges. Therefore, we assume that $G$ has at least one edge and therefore $\rho(H')> 0$.  Suppose by contradiction that $H'$ is a densest subgraph of $S(G)$ with $0 < \rho(H') \neq 1$ and that it is not a subdivision of a subgraph of $G$. We assume, without loss of generality, that $H'$ is connected (otherwise, consider each connected component separately). Then $H'$ must have at least one vertex $w$ which was introduced when subdividing an edge $\{u,v\}$ in $G$. Moreover, at least one such vertex $w$ must have degree one in $H'$, since otherwise, $H'$ would be a subdivision of a subgraph of $G$. Notice that the neighbor of $w$ in $H'$ is one of the original vertices of $G$. Suppose without loss of generality that this neighbor is $v$, and therefore $u$ is not in $H'$.
	We will now show that either $H'-{w}$ (i.e.the graph obtained by removing $w$ from $H'$) or $H'+{\{u,\{u,w\}\}}$ (i.e. the graph obtained by adding $u$ and its edge incident with $w$ to $H'$) has density greater than $H'$, that leads to a contradiction. Since $\rho(H')\neq 1$ we have two cases:
	
	\emph{Case $\rho(H') > 1$:} We note that $\rho(H'-{w}) = \frac{m_{H'}-1}{n_{H'}-1}$. 
	Then, using \cref{lem:fractions} (with~$c=d=1$) we obtain that $\rho(H'-{w}) = \frac{m_{H'}-1}{n_{H'}-1} > \frac{m_{H'}}{n_{H'}} = \rho(H')$.

	\emph{Case $\rho(H') < 1$:} We note that $\rho(H'+{\{u,\{u,w\}\}}) = \frac{m_{H'}+1}{n_{H'}+1}$. 
	Using \cref{lem:fractions} (with~$c=d=1$) we know that $\rho(H'+{\{u,\{u,w\}\}}) = \frac{m_{H'}+1}{n_{H'}+1} > \frac{m_{H'}}{n_{H'}}= \rho(H')$. 
\end{proof}

Together \cref{lem:subdivision_of_subgraph} and \cref{lem:densest_subgraph_is_subdivision_rho_not_1} tell us that if 
$\rho^*(S(G))\neq 1$ the subdivision of a densest subgraph of $G$ is a densest subgraph of $S(G)$. 

We now show that when $\rho^*(S(G))=1$ there exists a densest subgraph of $S(G)$ which is a subdivision of a densest subgraph of $G$. This will be done in two parts: First, we show that if $\rho^*(S(G))= 1$,then we can always find a densest subgraph of $S(G)$ that is a subdivision of some subgraph of $G$ (\cref{prop:densest_subgraph_is_subdivision_rho_1}). Second, we show that such a subgraph of $G$ must be a densest subgraph of $G$ (\cref{lem:density_1_THEN_subdivision_density_1}).

\begin{lemma}\label{prop:densest_subgraph_is_subdivision_rho_1}
	Let $H'$ be a densest subgraph of $S(G)$. If $\rho^*(S(G))=\rho(H')=1$, then there exists $H''\subseteq H'$ such that $\rho(H'')=\rho(H')=1$ and $H''=S(G')$ for some $G'\subseteq G$.
\end{lemma}
\begin{proof}
	Since $\rho^*(S(G))=1$ any connected densest subgraph of $S(G)$ contains a single cycle. Moreover, any cycle in $S(G)$ is a subdivision of a cycle in $G$. Assume that $H'$ is connected, otherwise treat each connected component separately. Denote by $C'$ the unique cycle in $H'$ and let $C$ be a cycle in $G$ such that $S(C)=C'$. Then setting $H''= C'$ and $G'=C$ we have: $H''\subseteq H'$ such that $\rho(H'')=1$ and $H''=S(G')$.
\end{proof}

The next lemma implies that when $\rho^*(S(G))=1$ the graph $G'$ of \cref{prop:densest_subgraph_is_subdivision_rho_1} must be a densest subgraph of $G$.

\begin{lemma}\label{lem:density_1_THEN_subdivision_density_1}
	A graph $G$ has density $\rho(G)=1$ if and only if $\rho(S(G))=1$.
\end{lemma}
\begin{proof}
	Assume without loss of generality that $G$ (respectively, $S(G)$) is connected, otherwise, treat each connected component separately.
	
	``$\Rightarrow:$'' Since $G$ is a graph of density $\rho(G)=1$,  $G$ must be a tree with one additional edge. In other words, $G$ is connected and contains a single cycle. Subdividing $G$ preserves the number of cycles in $G$ and the connection between any two vertices in $G$.%
    Thus $S(G)$ is a tree with one additional edge, and consequently, it has density $\rho(S(G))=1$. 

	``$\Leftarrow:$'' Let $G$ be a graph whose subdivision $S(G)$ has density $\rho(S(G))=1$. Then, $S(G)$ is a tree with one additional edge (i.e. it contains a single cycle). Since subdivision is cycle preserving, does not disconnect the graph and preserves the degree of vertices in $G$, $G$ must be connected and contain a single cycle. Thus, $G$ must be a tree plus one edge, which means $\rho(G)=1$.
\end{proof}

Armed with the previous lemmas we are ready to prove the following result:

\begin{lemma}\label{thm:subdivision_preserves_density}
    Let $G$ be a graph and $H\subseteq G$ a densest subgraph of $G$. Then $S(H)$ is a densest subgraph of $S(G)$. In other words, $\rho^*(S(G))=\rho(S(H))$.
\end{lemma}
\begin{proof}
	If $\rho^*(S(G))\neq 1$, by \cref{lem:densest_subgraph_is_subdivision_rho_not_1} the densest subgraph of $S(G)$ is a subivision of a subgraph of $G$ and by \cref{lem:subdivision_of_subgraph} any subgraph in~$S(G)$ that is a subdivision of a subgraph of~$G$ is at most as dense as~$S(H)$. So $S(H)$ must be a densest subgraph of $S(G)$. 

	If $\rho^*(S(G)) = 1$, by \cref{prop:densest_subgraph_is_subdivision_rho_1}, there exists a subdivision of $G$ which is a densest subgraph of $G$. Then, by \cref{lem:density_1_THEN_subdivision_density_1} and \cref{obs:density_of_subdivision} the densest subgraph of $G$ must have density~$1$. In other words, $H$ has density~$1$. Finally, by \cref{lem:density_1_THEN_subdivision_density_1}, $S(H)$ has density~$1$, which means, $S(H)$ is a densest subgraph of $S(G)$.
\end{proof}

Now we prove the main result of this section.

\begin{proof}[Proof of \cref{thm:hardness-densest-edge-bipartite}]
    We reduce from \densestEdgeDeletion{} on planar graphs such that $(m-k)/n> 1$ which is \NP-complete by \cref{cor:NP-completeness_G-F_bigger_1}.
    Given an instance $(G, k,\tau_\rho = p/q)$ of \densestEdgeDeletion{} with $p,q$ positive integers and $G$ such that $(m_G-k)/n_G>1$, create an instance $(G',k',\tau_\rho')$ of \textsc{\densestEdgeDeletion{}} follows: Set $G'$ to the subdivision of $G$, set $k'=k$, and set $\tau_\rho'=\frac{2p}{p+q}$. 
	Note that $G'$ is a planar bipartite graph and $(m_G-k)/n_G\geq 1$ implies $(m_{G'}-k)/n_{G'} > 1$ as $m_{G'} = 2m_G$ and $n_{G'} = m_G + n_G$.
	
	We claim that $(G, k,\tau_\rho)$ is a yes-instance of \densestEdgeDeletion{} if and only if $(G',k',\tau_\rho')$ is a yes-instance of \textsc{\densestEdgeDeletion{}}.
	
``$\Rightarrow:$'' Suppose $(G, k,\tau_\rho)$ is a yes-instance with solution set $F$. Let $H$ be a densest subgraph in $G-F$. By \cref{thm:subdivision_preserves_density}, $S(H)$ is a densest subgraph of $S(G-F)$. By \cref{obs:density_of_subdivision} we know that $\rho(S(H)) = 2m_H/(n_H + m_H)$. 
 Since $\frac{m_H}{n_H}  \leq \frac{p}{q}$ then $\rho(S(H)) = \frac{2 m_H}{n_H + m_H} \leq \frac{2p}{p + q}$.

	We now show that we can find a set $F'\subseteq E'$ of size $k$ so that a densest subgraph of $S(G-F)$ is also a densest subgraph of $G'-F'$.

	Recall that an edge $e$ in $G$ has an associated vertex $y$ in $S(G)$.

	We first build the set $F'$ of edges to remove from $G'=S(G)$ as follows: For each edge $e \in F$, arbitrarily pick one edge $e'$ of the two edges incident to vertex $y$ in $G'$, and add $e'$ to $F'$. Notice that vertex $y$ will be of degree $1$ in $G'-F'$, and that $|F'|=|F|=k$. 
	
	Moreover, we can obtain $S(G-F)$ by removing all vertices $y$ in $G'-F'$ associated with an edge $e\in F$. In other words: $S(G-F) = (G'-F')-\{y|y \text{ is a vertex associated to an edge }e\in F\}$. Since any such vertex $y$ is of degree 1, any subgraph of $G'-F'$ that includes $y$ is not a subdivision of a subgraph of $G$. This means that any subgraph of $G'-F'$ that is the subdivision of a subgraph of $G$, must be a subgraph of $S(G-F)$. We now use the fact that $G'$ is such that $\rho(G'-F')\geq 1$ (thus $\rho^*(G'-F')\geq 1$), for any $F'$ of size at most $k$, and we distinguish two cases:

	\emph{Case 1:} If $\rho^*(G'-F') > 1$ then no vertex of degree 1 can be part of a densest subgraph. So the densest subgraph of $G'-F'$ must be a subgraph of $S(G-F)$. Thus, $S(H)$ is a densest subgraph of $G'-F'$

	\emph{Case 2:} If $\rho^*(G'-F') = 1$ then by \cref{prop:densest_subgraph_is_subdivision_rho_1} there is a subdivision of a subgraph of~$G-F$ with density~$1$. By \cref{thm:subdivision_preserves_density} $S(H)$ is a densest subgraph of~$G'-F'$.

	Thus $S(H)$ is always a densest subgraph of~$G'-F'$. Since $|F'|=k=k'$ and $\rho(S(H))\leq \tau_\rho'$, we know that $(G',k',\tau_\rho')$ is a yes-instance for \textsc{\densestEdgeDeletion{}}.

	``$\Leftarrow:$'' If $(G', k',\tau_\rho')$ is a yes-instance with solution set $F'$, we can construct a set $F$ as follows: for each edge $e'$ in $F'$ put edge $e$ in $F$ if $e'$ is incident on vertex $y$ associated to edge $e$. Note, that we could have two edges $e_1',e_2'$ in $F'$ incident on the same vertex $y$ associated with edge $e$ in $G$, however,  $e$ would only appear once in $F$. Thus $|F|\leq|F'|=k'=k$.

	We will now prove by contradiction that $\rho^*(G-F)\leq  \frac{p}{q}$.
	Suppose that there is a subgraph $H$ of $G-F$ with density $\rho(H) > \frac{p}{q}$. Since $S(G-F)$ is a subgraph of $G'-F'$, we know that $S(H)$ is also a subgraph of $G'-F'$.  Since $\frac{m_H}{n_H} > \frac{p}{q}$ we have $ \rho(S(H))=\frac{2m_H}{n_H+m_H} > \frac{2p}{p+q}$, a contradiction, that is $\rho^*(G-F)\leq  \frac{p}{q}$.

\end{proof}
}

\begin{theorem}
	\label{thm:edgeDeletion-NP-hard-on-Split}
    \densestEdgeDeletion{} is \NP-complete on split graphs with~$\tau_\rho=3/4$.
\end{theorem}
\appendixproof{thm:edgeDeletion-NP-hard-on-Split}
{
\begin{proof}
	For the \NP-hardness we provide a polynomial-time reduction from \textsc{Exact Cover by 3 Sets}, also known as \textsc{X3C}, which is \NP-complete~\cite{garey:1979:computers_and_intractability}. 
 	First, recall \textsc{X3C}.
	\begin{problem}
		\problemtitle{\textsc{Exact Cover by 3 Sets} (X3C)}
		\probleminput{A set $X$ with $|X|=3q$ and a collection $C$ of $3$-element subsets of $X$.}
		\problemquestion{Is there a subset $C'$ of $C$ where every element of $X$ occurs in exactly one member of $C'$?  (Such a $C'$ is called an ``exact cover'' of $X$).}
	\end{problem}

	The reduction, illustrated in \cref{fig:NP-complete-edge-removal-split-graph-example}, is as follows: 
	Take an instance of $(X,C)$ of \textsc{X3C} with $X=\{x_1,x_2,\ldots, x_{3q}\}$ and $C=\{C_1,C_2,\ldots,C_t\}$.
	Create a vertex for each element in $X$ and a vertex for each $3$-element subset in $C$. Add an edge $\{x_i,C_j\}$ if $x_i$ is an element of $C_j$ in $(X,C)$. Next, create a set of  dummy vertices~$U=\{u_1,u_2,\ldots,u_{t-q}\}$ of size $t-q$ and add all edges between vertices in~$C\cup U$, making~$C \cup U$ a clique. 
	Then, for each dummy  vertex $u_i \in U$, create two degree~$1$-vertices $y_i,z_i$ whose incident edge is $\{u_i,y_i\}$ and $\{u_i,z_i\}$, respectively. Denote the set of all $y_i$ and $z_i$ vertices by $D$. 
	Call the graph obtained through this procedure $G$. Notice that the vertex set of $G$ is $V(G) = X\cup C\cup U \cup D$ of size $n=3q+t+3(t-q)= 4t$. Moreover, since $C \cup U$ is a clique and $X\cup D$ is an independent set, $G$ is a split graph.

	We claim that the instance $(X,C)$ of \textsc{X3C} is a yes-instance if and only if $(G,k=m-3t, \tau_\rho=\frac{3}{4})$ is a yes-instance of \densestEdgeDeletion{}.

	\begin{figure}
		\centering
			\begin{tikzpicture}[scale=0.80,node distance=8mm]%
			\def\xcoor{2}
			\def\c{4}
			\begin{scope}[xshift=-2cm,start chain=going right]
				\foreach \i in {1,2,3}{
					\node[on chain,circle,draw, fill, inner sep=2pt] (x-\i)[label=below:{$x_\i$}] {};
				}
				\node[on chain, circle, inner sep=2pt] (x-dots) {\ldots};
				\node[on chain, circle, draw,fill,inner sep=2pt, label=below:{$x_{3q}$}] (x-n) {};
			\end{scope}
			\begin{scope}[yshift=3cm,start chain=going right]
				\foreach \i in {1,2}{
					\node[on chain, circle,draw,fill,inner sep=2pt,label=above:{$C_\i$}] (c-\i){};
				}
				\node[on chain, circle,inner sep=2pt] (c-dots) {\ldots};
				\node[on chain, circle,draw,fill,inner sep=2pt, label=above:{$C_t$}] (c-m) {};
			\end{scope}
			\begin{scope}[yshift=3cm,xshift=7cm,start chain=going right]
				\node[on chain, circle,draw,fill,inner sep=2pt,label=above:{$u_1$}](u-1) {};
				\node[on chain, circle,draw,fill,inner sep=2pt,label=above:{$u_2$}](u-2) {};
				
				\node[on chain, circle,inner sep=2pt] (u-dots) {\ldots};
				\node[on chain, circle,draw,fill,inner sep=2pt, label=above:{$u_{t-q}$}] (u-m) {};
			\end{scope}
			\begin{scope}[xshift=5.5cm,start chain=going right]
				\node[on chain, circle,draw,fill,inner sep=2pt,label=below:{$y_1$}](y-1) {};
				\node[on chain, circle,draw,fill,inner sep=2pt,label=below:{$z_1$}](z-1) {};
				
				\node[on chain, circle,draw,fill,inner sep=2pt,label=below:{$y_2$}](y-2) {};
				\node[on chain, circle,draw,fill,inner sep=2pt,label=below:{$z_2$}](z-2) {};
				
				\node[on chain, circle,inner sep=2pt] (y-dots) {\ldots};
				
				\node[on chain, circle,draw,fill,inner sep=2pt,label=below:{$y_{t-q}$}](y-m) {};
				\node[on chain, circle,draw,fill,inner sep=2pt,label=below:{$z_{t-q}$}](z-m) {};
			\end{scope}
			\draw[thick](x-1)--(c-1);
			\draw[thick](x-2)--(c-1);
			\draw[thick](x-3)--(c-1);
			
			\draw[thick](x-2)--(c-2);
			\draw[thick](x-n)--(c-m);
			
			\draw[thick](y-1)--(u-1);
			\draw[thick](z-1)--(u-1);
			
			\draw[thick](y-2)--(u-2);
			\draw[thick](z-2)--(u-2);
			
			\draw[thick](y-m)--(u-m);
			\draw[thick](z-m)--(u-m);

			\draw[color=gray](x-2)--(c-dots);
			\draw[color = gray](x-2) edge[bend left=16] (c-dots);
			\draw[color = gray](x-2) edge[bend right=16] (c-dots);
			
			\draw[color=gray](x-3)--(c-dots);
			
			\draw[color=gray](x-dots)--(c-2);
			\draw[color = gray](x-dots) edge[bend left=16] (c-2);
			
			\draw[color = gray](x-dots)--(c-m);
			\draw[color = gray](x-dots) edge[bend left=16] (c-m);
			
			\draw[color=gray](x-n)--(c-dots);
			\draw[color = gray](x-n) edge[bend left=16] (c-dots);

			\node[inner sep = 19pt, very thick, italyGreen, draw,rounded corners, dashed, fit=(c-1) (c-2) (c-m),label=above:$C$](C) {};
			\node[inner sep = 12pt, very thick, italyRed, draw,rounded corners, dashed, fit=(x-1) (x-2) (x-n),label=below:$X$] (X){};
			\node[inner sep = 12pt, very thick, black!40, draw,rounded corners, dashed, fit=(y-1) (y-2) (y-m) (z-m),label=below:$D$] (D){};
			\node[inner sep = 12pt, very thick, black!40, draw,rounded corners, dashed, fit=(u-1) (u-2) (u-dots) (u-m),label=above:$U$] (U){};
			\node[inner sep = 16pt, thick,draw,rounded corners, fit=(C)(U)(c-1) (c-2) (c-m) (u-1) (u-2) (u-m),label=above:Clique] {};
			
			\end{tikzpicture}
		\caption{
				A schematic illustration of the reduction in \cref{thm:edgeDeletion-NP-hard-on-Split}. 
				The \textsc{X3C}  instance is encoded in the displayed graph by the vertices of $X$ (red) and $C$ (green). The edges between~$X$ and~$C$ encode which elements of $X$ belonging to the sets in $C$.
				The other vertices are dummy vertices to limit the number of vertices in~$C$ that can be connected to vertices in $X$.
			}
			\label{fig:NP-complete-edge-removal-split-graph-example}
		\end{figure}
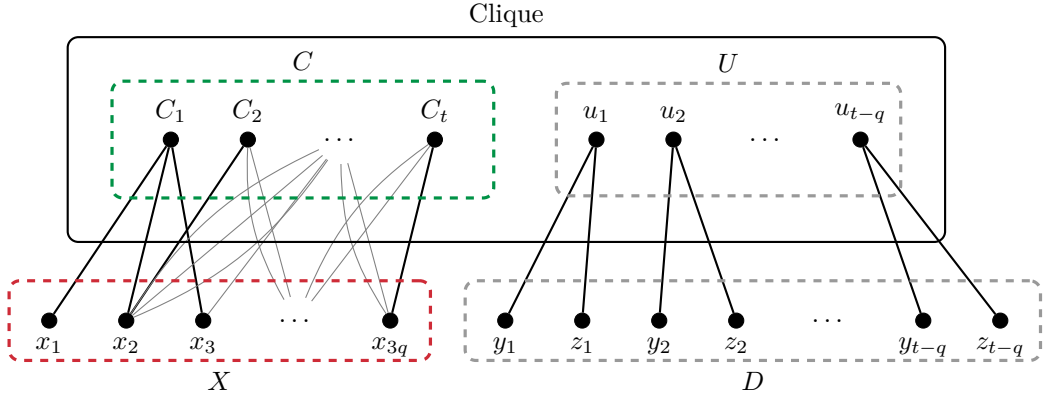

	``$\Rightarrow:$'' 
	Let $C'$ be an exact cover of $(X,C)$. We construct the following set of edges $B$: %
	For each $C_j \in C'$ include in $B$ the three edges that join $C_j$ to vertices in $X$, there are $3|C'|$ such edges. 
	For each $C_j \notin C'$ pick a distinct $u_i\in U$ and include in $B$  edges $\{u_i,C_j\}$, $\{y_i,u_i\}$ and $\{z_i,u_i\}$, there are $3(t-|C'|)$ such edges. 
	Thus, the set $B$ has size $3|C'|+3(t-|C'|)=3t$.
	Next, we define $F$ as the set of all edges in $G$ except the set $B$. The size of $F$ is of size $m-3t = k$. 
	Moreover, graph $G-F$,  that contains edges of $B$, is a forest of $K_{1,3}$. 
	Thus, $\rho^*(G-F )=\frac{3}{4}=\tau_\rho$.
	
	``$\Leftarrow:$'' Consider an  edge-set $F$, solution for the instance
	  $(G,k=m-3t,\tau_\rho=\frac{3}{4})$. 
	Since $G-F$ has $4t$ vertices and $3t$ edges and $\rho^*(G-F )\leq\frac{3}{4}$, it follows from \cref{lem:nice-props} (\ref{niceprop-trees-density}.) that $G-F$ is a forest of trees on 4 vertices, that is a perfect $\{K_{1,3},P_4\}$-packing of $G$. A \emph{perfect $\{K_{1,3},P_4\}$-packing of $G$} is a partition of $V(G)$ into sets $V_1,V_2,\ldots,V_t$ such that for all $i\in[t]$ the graph $G[V_i]$ is isomorphic to either $K_{1,3}$ or $P_4$.

	We prove in the following that any perfect $\{K_{1,3},P_4\}$-packing of $G$ must consist \emph{entirely} of $K_{1,3}$, and every $K_{1,3}$ in the packing has either a central vertex $u_i$ and leafs $y_i,z_i,C_j$ for some $i,j$, or it has central vertex $C_j$ and leaves $x_i,x_\ell, x_p$ for some $j,i,\ell,p$:
	Consider a perfect $\{K_{1,3},P_4\}$-packing of $G$. 
	Now consider a pair $y_i,z_i$ of vertices in $D$ which have $u_i$ as common neighbor: $z_i,y_i$ and $u_i$ must belong to the same connected component. 
	The fourth vertex in that connected component, call it $w_i$, must have $u_i$ as neighbor. 
	Thus, $u_i$ has degree $3$ in the packing, which implies $u_i,y_i,z_i, w_i$ form a $K_{1,3}$ with $u_i$ as central vertex. 
	This is true for every $i=1,2,\ldots, t-q$. 
	By construction, there are as many $y_i,z_i$ pairs as there are vertices $u_i$, so the vertices $z_i$ cannot be elements of $U$. Thus, every vertex $z$ is a distinct element of $C$.
	Thus, there are $q$ vertices in $C$ which are not in a $K_{1,3}$ with some vertex $u_i\in U$. 
	Denote by $R$ the set of these $q$ vertices. 
	There must be a perfect $\{K_{1,3},P_4\}$-packing of the vertices in $R\cup X$.
	We show by contradiction that no two vertices in $R$ can be packed together: Suppose there are two remaining vertices $C_j,C_\ell$ that are in the same connected component in the packing. 
	Then, there are at most $q - 2$ vertices in $R$ and at least $3q-2$ vertices in $X$ for which there must be a perfect $\{K_{1,3},P_4\}$-packing. 
	Notice that each vertex from $X$ must be in the same connected component as some vertex from $R$ because $X$ is an independent set. 
	But each vertex in $R$ has at most $3$ neighbors in $X$. 
	Thus, the at most $q-2$ vertices from $R$ can be packed with at most $3q-6$ vertices in $X$. 
	But there are at least $3q-2$ vertices in $X$, so there are at least $4$ vertices in $X$ left over after packing all the other vertices (and they cannot be packed together). Which is a contradiction.
	Since no two vertices of $R$ can be packed together, and $X$ is an independent set, then any perfect $\{K_{1,3},P_4\}$-packing of the vertices in $X\cup R$ must be a perfect $K_{1,3}$-packing. Moreover each $K_{1,3}$ in this $K_{1,3}$-packing has a single vertex form $C$ and three vertices from $X$.
	In particular, this is true for $G-F$. Thus the set $C'$ consisting of all $C_j \in C$ which are in the same connected component as some vertex of $X$ in $G-F$ is an exact three cover for $(X,C)$.
\end{proof}
}

\subsection{Parameterized Complexity Results}\label{sec:edge-del-parameterized}
\subparagraph*{FPT wrt.\ Vertex Cover Number.}
The vertex cover number of a graph denotes the size of a smallest vertex cover, i.\,e., the size of a set of vertices whose removal results in an edge-less graph.
Although this parameter is relatively large, we still need to rely on integer linear programming in our next algorithm.

\begin{theorem}\label{thm:edge-deletion-FPT-vertex-cover-number}
	\densestEdgeDeletion{} can be solved in~$2^{O(\ell 2^{2\ell})}+O(m+n)$ time where~$\ell$ is the vertex cover number.
\end{theorem}
\begin{proof}
    Consider an instance $(G,k,\tau_\rho)$ of \densestEdgeDeletion{}. 

	We provide an algorithm that first computes a minimum vertex cover~$C \subseteq V(G)$, $|C| = \ell$, in~$O(2^{\ell} + n + m)$ time~\cite{CFKL2015}.
	Then it computes edge deletions within~$G[C]$ and incident to~$S = V(G) \setminus C$ with an ILP, i.\,e., computes $F$. 
	To this end, we divide the vertices in $S$ into at most $2^\ell$ classes $I_1, \ldots, I_{2^{\ell}}$, where two vertices are in the same class if and only if they have the same neighbors.
	Hence, we can define for a class~$I_i$ the neighborhood~$N(I_i) = N(v)$ for~$v \in I_i$.
	We denote with~$|I_i|$ the number of vertices in the class~$I_i$. 
	The usefulness of these classes hinges on the fact that at least one densest subgraph $G'$ of $G$ is such that it either contains all the vertices of a class $I_i$ or none.
	This is a consequence of \cref{lem:fractions}: if removing (adding) a vertex of $I_i$ from a subgraph of $G$ increases its density, then removing (adding) any other vertices from $I_i$ must do the same, as they are twins and non-adjacent.
 
	We say a class~$I_j$ is \emph{obtainable} from a class~$I_i$ if~$N(I_j) \subseteq N(I_i)$, that is, by deleting some edges a vertex from~$I_i$ can get into class~$I_j$.
	Note that~$I_j$ is obtainable from~$I_j$.
	Denote with~$\ob(I_i)$ all classes obtainable from~$I_i$, formally, $\ob(I_i) = \{I_j \mid N(I_j) \subseteq N(I_i)\}$.
	Similarly, we denote with~$\ob^{-1}(I_j)$ all classes~$I_i$ so that~$I_j$ is obtainable from~$I_i$, formally, $\ob^{-1}(I_j) = \{I_i \mid N(I_j) \subseteq N(I_i)\}$.
	If~$I_j$ is obtainable from~$I_i$, then we set~$\cost(i \to j)= |N(I_i) \setminus N(I_j)|$ to be the number of edges that need to be deleted from a vertex~$v\in I_i$ to make it a vertex in~$I_j$.

	We now give our ILP. 
	To handle edges with one endpoint in the independent set~$S$ we do the following:
    For each pair of classes~$I_i,I_j$, we add a variable~$x_{i\to j}$ whose purpose is to denote how many vertices from class~$I_i$ will end up in class~$I_j$ by deleting edges.
	For convenience, we further have a variable~$y_j=\sum_{I_i \in \ob^{-1}(I_j)} x_{i\to j}$ denoting the total number of vertices that end up in a class~$I_j$ after deleting edges.
    To ensure correctness we require that no vertex from class~$I_i$ gets transformed into a vertex from a class not obtainable from~$I_i$ which we represent with the constraint~$\sum_{I_j \in \ob(I_i)} x_{i\to j}  = |I_i|$ for all~$ i\in [2^\ell]$.
    Moreover, both $x_{i\to j}$ and $y_j$ must be integers for all $i,j\in[2^\ell]$.

    To handle edges within $C$ we do the following:
    For each edge $e\in E(C)$ we create a binary variable $z_e$, where $z_e=1$ if $e$ survives and $0$ if it gets deleted.
	Again for convenience, for a subset~$C' \subseteq C$ we denote with~$m_{C'}=\sum_{e\in E(C')}z_e$ the number of edges within~$C'$ that survive.
    The ILP is as follows:

	\begin{align}
		\text{Minimize} && \sum_{j = 1}^{2^\ell}\sum_{I_i \in \ob(I_j)} \cost(i\to j) \cdot x_{i \to j} &+\sum_{e\in E(C)}(1-z_e)\\
		\text{such that}&& \sum_{I_j \in \ob(I_i)} x_{i \to j} & = |I_i| & \forall i\in [2^\ell] \\
						&& \sum_{I_i \in \ob^{-1}(I_j)} x_{i\to j} & = y_j & \forall j\in [2^\ell] \\
                        && \sum_{e\in E(C')}z_e & =m_{C'} & \forall C'\subseteq C\\
						&& m_{C'} + \sum_{i \in \mathcal{I}} y_i|N(I_i) \cap C'| & \le \tau_\rho (|C'| + \sum_{i \in \mathcal{I}} y_i)  & \forall C' \subseteq C, \forall \mathcal{I} \subseteq [2^\ell]\label{eq:densiy-check}\\
						&& x_{i\to j}, y_{i} & \in \{0,1,2,\ldots\} & \forall i,j\in [2^\ell]\\
                        && z_e & \in\{0,1\} & \forall e\in E(C)
	\end{align}
 
    To prove correctness it remains to show that the ILP admits a solution such that the objective value is at most $k$ 
    if and only if there exist an edge deletion set $F$ of size at most $k$ such that $\rho^*(G-F)\leq \tau_\rho$.

    Firstly, suppose there exists a feasible solution to the ILP that achieves an objective value of at most~$k$. 
    Because the vertices in~$I_i$ are indistinguishable, the set of variables $\{x_{i\to j} \mid j\in [2^\ell]\}$ uniquely determines the set of edges incident to~$I_i$ that need to go into~$F$ for each class~$I_i, i\in [2^\ell]$. 
    The other edges in $F$ can be directly read off the solution $E(C)\cap F = \{z_e \mid z_e=0\}$.
    A densest subgraph~$G'$ of~$G$ contains some vertices from~$C$ and from~$S$.
    If~$G'$ does not contain all vertices from some class~$I_i$, then either adding or removing all vertices from~$I_i$ will give a graph~$G''$ that is as least as dense as~$G'$ (cf.~\cref{lem:fractions}).
    Inequality~\eqref{eq:densiy-check} ensures that all such possible subgraphs~$G''$ have density at most~$\tau_\rho$.
    Therefore, this constraint guarantees that all subgraphs of $G-F$ have density at most $\tau_\rho$.
    In the objective function, $\sum_{e\in E(C)}(1-z_e)$ counts the number of edges that are deleted from $C$, whilst $\sum_{j = 1}^{2^\ell}\sum_{I_i \in \ob(I_j)} \cost(i\to j) \cdot x_{i \to j}$ counts the number of edges that are deleted from $E(C,S)$. Since the feasible solution achieves an objective value of at most $k$, we have that $|F|\leq k$.
    Thus, $F$ as built above is the desired edge deletion set of size at most $k$.

    Secondly, consider an edge subset $F\subseteq E(G)$ of size at most $k$ such that $\rho^*(G-F)\leq \tau_\rho$. We can assign values to $x_{i\to j}$ by reading the number of edges in $F$ that are incident to each vertex $v$ in class $I_i$. The first and second constraint will be satisfied by construction. The third constraint will be satisfied because $\rho^*(G-F)\leq \tau_\rho$. And since $|F|\leq k$ the minimum of the objective function will have value at most $k$.

    The ILP has at most $2^{2\ell}$ variables of type $x_{i\to j}$ %
    and at most $\binom{\ell}{2}$ of type $z_e$. The linear program can be expressed without variables of types $y_j$ and $m_{C'}$ simply by substituting their expressions into the other constraints. 
    This yields a total of $O(2^{2\ell+1})$ variables. In addition, we have $O(2^\ell \cdot 2^{2^\ell})$ constraints.
    Since an ILP instance $\mathcal{I}$ on $p$ variables can be solved in time $O(p^{2.5p+o(p)}\cdot |\mathcal{I}|)$ \cite{L1983, K1987,FT1987} we obtain a total (FPT) running time of $2^{O(\ell 2^{2\ell})}+O(m+n)$ for our algorithm. We note that one can achieve a running time of $\ell^{O(2^{2\ell})} + O(n+m)$ by using a randomized algorithm that solves ILPs on $p$ variables in time $\log(2p)^{O(p)}$\cite{RR2023}. 
\end{proof}

\subparagraph*{W[1]-Hardness wrt.\ feedback edge number and solution size.}

We next prove that \densestEdgeDeletion{} is W[1]-hard with respect to the combined parameter solution size and feedback edge number.
To this end, we use the construction of Enciso et al.~\cite{EFGK2009}, who reduced \textsc{Multicolored Clique} to \textsc{Equitable Connected Partition}.
These problems are defined as follows.

\begin{problem}
	\problemtitle{\textsc{Multicolored Clique}}
	\probleminput{An undirected graph~$G$ properly colored with~$\ell$ colors.}
	\problemquestion{Does~$G$ contain a clique on~$\ell$ vertices?}
\end{problem}

\begin{problem}
	\problemtitle{\textsc{Equitable Connected Partition}}
	\probleminput{An undirected graph~$G$ and a positive integer~$c$.}
	\problemquestion{Is there a partition of~$G$ into~$c$ equally sized connected components, that is, a partition of~$V(G)$ into~$V_1, \ldots, V_c$ so that each~$G[V_i]$ is connected and~$||V_i| - |V_j|| \le 1$ for each~$i,j \in [c]$?
	}
\end{problem}

Note that since~$G$ is properly colored any clique  in~$G$ is multicolored, that is, contains at most one vertex per color.

As a byproduct, we also obtain W[1]-hardness for $T_{h+1}$-\textsc{Free Edge Deletion} parameterized by the combined parameter feedback edge number and solution size. %
This confirms a conjecture by Enright and Meeks~\cite{EM2018}: $T_{h+1}$-\textsc{Free Edge Deletion} is indeed W[1]-hard with respect to treewidth.

\subparagraph*{Connection between problems.} 
Let us briefly discuss the connections between the three problems to see why the construction of Enciso et al.~\cite{EFGK2009} works for all three:
Consider a cycle on~$3n$ vertices. 
Requiring~$c = 3$ for \textsc{Equitable Connected Partition} guarantees one solution (up to symmetry): 
delete 3 edges so that 3 paths on~$n$ vertices remain.
Similarly, requiring~$k = 3$ and~$h = n$ for $T_{h+1}$-\textsc{Free Edge Deletion} guarantees the same solution.
The same applies for \densestEdgeDeletion{} with~$\tau_\rho = (n-1)/n$ and~$k = 3$.

The construction of Enciso et al.~\cite{EFGK2009} will enforce the following: 
Any connected component in a solution will be a tree. 
Hence, we can control its maximum size via the density~$\tau_\rho$ or directly via~$h$.
Combining the budget~$k$ with the size constraint ensures that the minimum number of vertices per component is equal to the maximum number of vertices per component.
Thus, for the constructed graph all three problems will ask essentially for the same solution.

\begin{theorem}\label{thm:edge-del-w-hard}
	\densestEdgeDeletion{} and $T_{h+1}$-\textsc{Free Edge Deletion} are W[1]-hard with respect to the combined parameter solution size and feedback edge number.
\end{theorem}

\begin{proof}
We use the same reduction as Enciso et al.~\cite{EFGK2009}. 
Let~$(G,\ell)$ be an instance of \textsc{Multicolored Clique}, which is W[1]-hard with respect to~$\ell$~\cite{FHRV09}.
Assume without loss of generality that~$G$ contains exactly $n/\ell$ vertices of each color $i\in [\ell]$. %
(We can add isolated vertices of the appropriate color to~$G$ to meet this demand.)
Denote with $\lambda\colon E(G) \to [m]$ a bijection that assigns each edge in~$G$ a unique integer from~$[m]$.
Further, denote with~$v^i_1, \ldots, v^i_{n/\ell}$ the vertices of color~$i$ in~$G$.

Construct instances~$(H,k,\tau_\rho)$ of \densestEdgeDeletion{} and~$(H,k,h)$ of $T_{h+1}$-\textsc{Free Edge Deletion} as follows.
Let~$\alpha$ and~$\beta$ be the smallest integers satisfying~$\beta > 2m + 10$ and~$\alpha > \beta \cdot n/\ell + 2m + 10$.
Then set~$h = \alpha + \beta \cdot n/\ell + m + 1$ which will be the maximum number of vertices any subtree in the resulting graph~$H-F$ is allowed to have.
To ensure this, set~$\tau_\rho = 1 - 1/h$.

The basic building block for~$H$ is the so-called \emph{anchor}.
For~$q \ge 1$, a $q$-anchor is a vertex adjacent to~$q-1$ many vertices of degree one.
The idea is that we cannot afford to cut off single vertices, thus all anchors in the constructed graph stay intact.
Hence, a $q$-anchor acts as a single vertex with weight~$q$ (the vertex contributes~$q$ to the size of its connected component).
We use anchors in the \emph{choice} gadget (see \cref{fig:w[1]-hard-edge-removal-choice-gadget} for an illustration):
	\tikzstyle{alpha-anchor}=[diamond,draw,fill,inner sep=2pt]
	\tikzstyle{beta-anchor}=[circle,draw,fill,inner sep=2pt,minimum size=2pt]
	\tikzstyle{betaAnchors} = [black!80,decoration={markings,mark=between positions 0.25 and 0.75 step 0.25 with { \draw (0pt,0pt) -- (0pt,-15pt);}},postaction={decorate}]
	\tikzstyle{Edges} = [black!80,decoration={markings,mark=between positions 0.1 and 0.95 step 0.05 with { \draw (0pt,0pt) -- (0pt,-4pt);}},postaction={decorate}]
	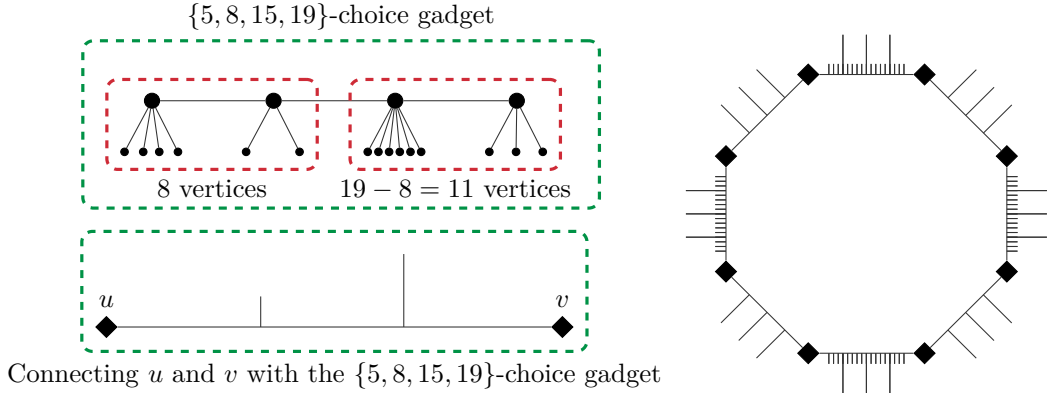
\begin{figure}[t]
		\centering
		
		\begin{tikzpicture}[scale = 1]
			\def\pathWidth{4}
			\node[alpha-anchor,label=above:$u$] at (- 0.75*\pathWidth,-3) (aa) {};
			\node[alpha-anchor,label=above:$v$] at (0.75*\pathWidth,-3) (bb) {};
			\foreach[count=\tt] \t in {0.2,0.4,0.6,0.8} 
			{
				\node[beta-anchor] at (-\pathWidth + \t * 2 * \pathWidth,0) (a\tt) {};
			}
			\foreach[count=\tt] \t in {2,3,4} 
			{
				\draw[black!80] (a\tt) -- (a\t);
			}
			\draw[black!80] (aa) -- (bb) node foreach[count=\tt] \t in {0.33,0.66} [pos=\t] (aa\tt) {};
			\foreach[count=\tt] \i in {4,2,6,3}
			{
				\foreach \j in {1,...,\i}
				{
					\pgfmathtruncatemacro{\pos}{(20*\j-20)/(\i-1) - 13}
					\node[circle,draw,fill,inner sep=1pt,minimum size=1pt,below left= 16 pt and \pos pt of a\tt] (a\tt-\j) {};
					\draw[black!80] (a\tt) -- (a\tt-\j);
				}
			}
			\foreach[count=\tt] \i in {2,6}
			{
				\node[above= 4*\i pt of aa\tt.center] (aa\tt-1) {};
				\draw[black!80] (aa\tt.center) -- (aa\tt-1.center);
			}
			\node[inner sep = 5pt,very thick, italyRed, draw,rounded corners, dashed,fit=(a1-4) (a2) (a2-1),label=below:{$8$ vertices}] {};
			\node[inner sep = 5pt,very thick, italyRed, draw,rounded corners, dashed,fit=(a3-6) (a3) (a4-1),label=below:{$19-8=11$ vertices}] (label-vertex) {};
			\node[inner sep = 14pt,very thick, italyGreen, draw,rounded corners, dashed,fit=(a1-4) (a4) (label-vertex),label=above:{$\{5,8,15,19\}$-choice gadget}] {};
			\node[inner sep = 5pt,very thick, italyGreen, draw,rounded corners, dashed,fit=(aa) (bb) (aa2-1),label=below:{Connecting~$u$ and~$v$ with the $\{5,8,15,19\}$-choice gadget}] {};

			\begin{scope}[xshift=7cm,yshift=-1.5cm]
				\def\k{8}
				\def\radius{2}
				\pgfmathtruncatemacro{\kk}{\k - 1}
				\foreach \i in {1,...,\k}
				{
					\node[alpha-anchor] at ({\i * 360 / \k - 180 / \k}:\radius) (v-\i) {};
				}
				\foreach[count=\xi from 2] \i in {1,...,\kk}
				{
					\pgfmathsetmacro\Edges{iseven(\i) ? "Edges": ""}
					\draw[betaAnchors,\Edges] (v-\i) -- (v-\xi);
				}
				\draw[betaAnchors,Edges] (v-\k) -- (v-1);
			\end{scope}
		\end{tikzpicture}
		\caption{
			\emph{Left above:} A $\{5,8,15,19\}$-choice gadget consisting of four anchors. The two red dashed boxed indicate a possible split into 8 and 11 vertices.
			\emph{Left below:} A more compact representation (used in further figures) of the same gadget used to connect~$u$ and~$v$.
			The length of the vertical lines correspond to the number of degree-one neighbors of the corresponding anchor; we omit the line for the first and last anchor.
			The diamond shaped vertices indicate $\alpha$-anchors. Each choice-gadget in the construction connects two $\alpha$-anchors as visualized above.
			\emph{Right:} A gadget constructed for each color. It consists of~$2(\ell-1)$ many $\alpha$-anchors (so $\ell=5$ in the example) that are connected in a cycle via choice-gadgets (in the example there are 4 vertices per color). 
		}
		\label{fig:w[1]-hard-edge-removal-choice-gadget}
	\end{figure}
Let~$A = \{a_1, \ldots, a_q\}$ be a set of integers so that~$1 \le a_i < a_j$ for all~$1 \le i < j$.
An~$A$-\emph{choice} is a path on~$p$ vertices~$u_1, \ldots, u_p$ where vertex~$u_i$, $i\in[p]$, is the center of an~$(a_i - a_{i-1})$-anchor ($u_1$ is the center of an~$a_1$-anchor). 
Note that an $A$-choice has exactly $a_p$ many vertices ($a_p$ is the maximum of $A$), and that removing the edge~$\{u_i,u_{i+1}\}$ splits the $A$-choice gadget into two connected components with~$a_i$ and~$a_q-a_i$ vertices respectively.
We say we \emph{cut} the $A$-choice at~$a_i$, $i \in [|A|-1]$, to indicate the removal of the edge~$\{u_i,u_{i+1}\}$.
To \emph{connect} two vertices~$u$ and~$v$ by an $A$-choice means to merge the first and last anchor of the $A$-choice gadget with~$u$ and~$v$, respectively.
Merging a vertex~$v$ with an $q$-anchor means to remove the anchor, add~$q$ degree-one vertices adjacent only to~$v$ and make~$v$ adjacent to all vertices the center of the anchor was adjacent to.
In other words, identify~$v$ with the center of the anchor and add one vertex only adjacent to~$v$ (to ensure the overall number of vertices stays the same).
Note that connecting two vertices~$u$ and~$v$ by an $A$-choice still leaves $|A|-1$ many possible cuts of the~$A$-choice.
We only connect center vertices of $\alpha$-anchors by choice gadgets, thus enforcing a cut in each choice gadget.

For each of the~$\ell$ colors in~$G$ add~$2(\ell-1)$ many $\alpha$-anchors to the initially empty graph~$H$.
Note that~$\alpha$ is sufficiently large to ensure that no two anchors can be in the same connected component, that is, $2\alpha > h$.
Denote with~$N^i_j,P^i_j$, $j \in [\ell]\setminus\{i\}$, the center vertices of the anchors for color~$i\in[\ell]$.
Set~$A_V = \{1\} \cup \{p\beta \mid p \in [n/\ell]\}$ containing~$1 + n/\ell$ elements, thus accommodating one cut for each vertex of color~$i$.
Let~$j'$ be the ``successor index'' of~$j$, formally:

\[j' = \begin{cases}
		1, & \text{if } (j=\ell \land i\neq 1) \lor (j+1= i=\ell) \\
  		2, & \text{if } j=\ell \land 1=i \\
		j+1, & \text{if } j+1 \neq i \land j+1 \le \ell \\
		j+2, & \text{if } j+1 = i \land j+2 \le \ell
      \end{cases}
\]
For each~$j \in [\ell]\setminus\{i\}$, connect~$N^i_j$ and~$P^i_{j'}$ by an $A_V$-choice gadget.
Set~$A^i_E = \{(p-1)\beta + \lambda(\{u,v^i_p\}) \mid p \in [n/\ell] \land u\in V(G) \land \{u,v^i_p\} \in E(G)\} \cup \{\beta n/\ell + m + 1\}$ containing one element for each edge incident to each vertex of color~$i$ and a large final element to make sure each $A^i_E$-choice has exactly $\beta n/\ell + m + 1$ many vertices.
Next, connect~$P^i_j$ and~$N^i_{j}$ by an $A^i_E$-choice gadget.
For~$i,j \in [\ell], i\neq j$ connect~$P^i_j$ and~$N^j_i$ by an $[m+1]$-choice gadget (leaving one $m$ cut possibilities), see \Cref{fig:w[1]-hard-edge-removal-edge-encoding} for an illustration and some intuition.
Finally, set~$k = 2(\ell-1)\ell + 2\binom{\ell}{2} = 3(\ell-1)\ell$.

	\begin{figure}[t!]
		\centering
		\begin{tikzpicture}[scale = 1]
			\def\m{16}
			\fill[fill=black!10,rounded corners] (-6,2) rectangle (6,-0.2);
			\fill[fill=black!10,rounded corners] (-6,-5) rectangle (6,-2.8);
			\foreach[count=\i] \x / \y / \txt in { {-2.2}/1/-1, {-1}/0/{}, 1/0/{}, {2.2}/1/{+1}}
			{
				\pgfmathtruncatemacro{\ii}{\i / 2}
				\pgfmathsetmacro\VertexName{iseven(\i) ? "N": "P"}
				\node[alpha-anchor,label=above:{$\VertexName^i_{j\txt}$}] at (\x * 2.5,\y) (a\i) {};
				\node[alpha-anchor,label=below:{$\VertexName^j_{i\txt}$}] at (-\x* 2.5,-\y - 3) (b\i) {};
			}
			\draw[white] (a2) -- (a3) node foreach \t in {0,...,16} [pos= \t/16] (aa-\t) {};
			\draw[white] (b2) -- (b3) node foreach \t in {0,...,16} [pos= \t/16] (bb-\t) {};

			\node[] at (0,1.5) () {part of gadget for color~$i$};
			\node[] at (0,-4.5) () {part of gadget for color~$j$};
			
			\foreach \i in {1,2,3}
			{
				\pgfmathtruncatemacro{\ii}{\i + 1}
				\draw[betaAnchors] (a\ii) -- (a\i);
				\draw[betaAnchors] (b\ii) -- (b\i);
			}
			
			\pgfmathtruncatemacro{\mm}{\m + 1}
			\draw[black!80] (a2) -- (b3) node foreach \t in {1,...,\m} [pos= \t / \mm] (ab-\t) {};
			\draw[black!80] (b2) -- (a3) node foreach \t in {1,...,\m} [pos= \t / \mm] (ba-\t) {};
			\foreach \i in {1,...,\m}
			{
				\node[left= 0pt of ab-\i.center] (a-\i-l) {};
				\node[right= 0pt of ba-\i.center] (b-\i-r) {};
				\draw[black!80] (a-\i-l.center) -- (ab-\i.center);
				\draw[black!80] (b-\i-r.center) -- (ba-\i.center);
			}
			\foreach \i in {2,3,6,9,10,11,14}
			{
				\node[above= 2pt of aa-\i.center] (aa-\i-b) {};
				\draw[black!80] (aa-\i-b.center) -- (aa-\i.center);
				\node[below= 3pt of bb-\i.center] (bb-\i-a) {};
				\draw[black!80] (bb-\i-a.center) -- (bb-\i.center);
			}
			
			\node[below left=-2pt and 10pt of ab-5.center] (aL1) {};
			\node[below right=-2pt and 10pt of ab-5.center] (aL2) {};
			\node[below left=5pt and 2pt of aa-10.center] (aL3) {};
			\node[above left=14pt and 2pt of aa-10.center] (aL4) {};
			\node[above left=25pt and 20pt of a2.center] (aL5) {};
			\node[left=30pt of a2.center] (aL6) {};
			\draw[italyRed,thick,densely dotted,rounded corners] (aL1.center) -- (aL2.center) -- (aL3.center) -- (aL4.center) -- (aL5.center) -- (aL6.center) -- cycle;

			\node[above left=-1pt and 10pt of ba-5.center] (aR1) {};
			\node[above right=-1pt and 10pt of ba-5.center] (aR2) {};
			\node[below left=20pt and 24pt of a4.center] (aR3) {};
			\node[above left=2pt and 34pt of a4.center] (aR4) {};
			\node[above left=14pt and 0pt of aa-10.center] (aR5) {};
			\node[below left=5pt and 0pt of aa-10.center] (aR6) {};
			\draw[italyGreen,thick,densely dotted,rounded corners] (aR1.center) -- (aR2.center) -- (aR3.center) -- (aR4.center) -- (aR5.center) -- (aR6.center) -- cycle;

			\node[above left=-2pt and 10pt of ba-5.center] (bR1) {};
			\node[above right=-2pt and 10pt of ba-5.center] (bR2) {};
			\node[above left=7pt and 7pt of b1.center] (bR3) {};
			\node[below left=10pt and 14pt of b1.center] (bR4) {};
			\node[below left=14pt and 0pt of bb-2.center] (bR5) {};
			\node[above left=5pt and 0pt of bb-2.center] (bR6) {};
			\draw[italyRed,thick,densely dotted,rounded corners] (bR1.center) -- (bR2.center) -- (bR3.center) -- (bR4.center) -- (bR5.center) -- (bR6.center) -- cycle;

			\node[below left=-0pt and 10pt of ab-5.center] (bL1) {};
			\node[below right=-0pt and 10pt of ab-5.center] (bL2) {};
			\node[above left=5pt and 2pt of bb-2.center] (bL3) {};
			\node[below left=14pt and 2pt of bb-2.center] (bL4) {};
			\node[below left=20pt and 4pt of b3.center] (bL5) {};
			\node[above left=2pt and 14pt of b3.center] (bL6) {};
			\draw[italyGreen,thick,densely dotted,rounded corners] (bL1.center) -- (bL2.center) -- (bL3.center) -- (bL4.center) -- (bL5.center) -- (bL6.center) -- cycle;
			
			\node[color=black!70] (v-choiceR) at (5,-1.5) {vertex selection};
			\draw[->,black!30] (v-choiceR) to (aR3);
			\draw[->,black!30] (v-choiceR) to (bR3);

			\node[color=black!70] (v-choiceL) at (-5,-1.5) {vertex selection};
			\draw[->,black!30] (v-choiceL) to (aL6);
			\draw[->,black!30] (v-choiceL) to (bL6);

			\node[color=black!70] (e-choice) at (0.4,-1.5) {edge selection};
			\foreach \u in {aR6,bR1,bL3,aL2}
			{
				\draw[->,black!30] (e-choice) to (\u);
			}
		\end{tikzpicture}
		\caption{
			Illustration of the representation of edges of the original graph in the constructed graph.
			The top part (highlighted in gray) indicates part of the gadget (cycle) created for color~$i$; correspondingly on the bottom for color~$j$.
			The intuition for the reduction is as follows.
			Consider a solution for~$H$, that is, a set of edges~$F$ so that~$\rho^*(H-F) \le \tau_\rho$.
 			Some connected components in~$H-F$ are indicated by dashed lines enclosing parts of the vertices.
 			The size constraint for each connected component (enforced by the density threshold~$\tau_\rho < 1$) ensures that each connected component of~$H-F$ contains at most one $\alpha$-anchor (indicated by diamond-shaped vertices), at most $n/\ell$ many $\beta$-anchors (indicated by the longer vertical lines in the choice-gadgets, here~$n/\ell = 4$), and an additional~$m+1$ vertices.
 			In the vertex-selection part (the $A_V$-choice gadgets that only contains $\beta$-anchors) there are~$n/\ell$ many choices to cut; each corresponding to selecting one of the~$n/\ell$ vertices of the respective color. 
 			The bound of at most $n/\ell$ many $\beta$-anchors per connected component enforces the same choice throughout the color gadget (see right side of \cref{fig:w[1]-hard-edge-removal-choice-gadget}).
 			The edge selection follows the same idea as the vertex selection for each color (the number are just smaller):
 			The $[m+1]$-choice gadgets between~$N^i_j$ and~$P^j_i$ as well as between~$P^i_j$ and~$N^j_i$ imply that~$2m+2$ vertices need to be distributed to the connected components around~$N^i_j,P^j_i,P^i_j,N^j_i$.
 			Moreover, the vertex pairs~$\{N^i_j,P^i_j\}$ and~$\{N^j_i,P^j_i\}$ are connected via an $A^i_E$- resp. $A^j_E$-choice gadget. 
 			Both of these gadgets contain~$n/\ell + m + 1$ vertices. 
 			Thus, $4m+4$ vertices needs to be distributed to the connected components around~$N^i_j,P^j_i,P^i_j,N^j_i$.
 			By construction, this requires a cut at an anchor representing the same edge in all four connections.
 			This implies that the selected vertices in the color gadgets are adjacent in the original graph.
 			As one vertex needs to be selected per color a clique needs to be selected.
		}
		\label{fig:w[1]-hard-edge-removal-edge-encoding}
	\end{figure}
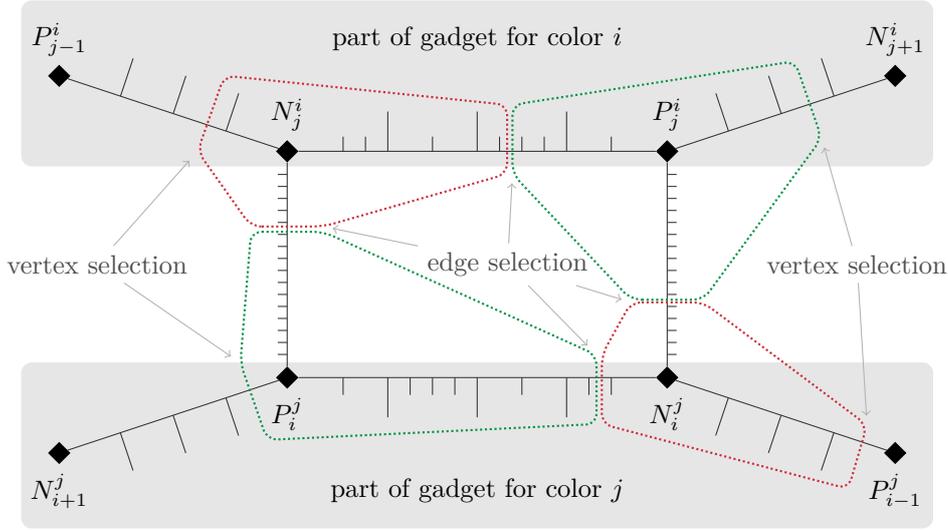

Observe that removing in~$H$ the degree-one vertices, we are left with~$\ell$ cycles that are pairwise connected by two paths. 
Thus, the feedback edge number of the constructed graph~$H$ is~$\ell+2\binom{\ell}{2}-(\ell-1) = 2\binom{\ell}{2}+1$.

We show that any solution for the instance~$(H,k,\tau_\rho)$ creates connected components of exactly the same size.
Let~$F$ be a solution to the constructed \densestEdgeDeletion{} instance~$(H,k,\tau_\rho)$, that is, $|F| \le k$ and the densest subgraph in~$G-F$ has density at most~$\tau_\rho$.
Observe that each ``large'' $\alpha$-anchor is in a different connected component in~$H-F$ as~$1 / (1-\tau_\rho) = h < 2\alpha$.
Thus, $F$ contains at least one edge of each choice gadget as each choice gadget connects two $\alpha$-anchors.
Since there are~$k$ choice gadgets, it follows that~$F$ contains exactly one edge from each choice gadget.
Thus, $H-F$ contains exactly~$2(\ell-1)\ell$ many connected components.
Note that~$H$ consists of~$2(\ell-1)\ell$ many $\alpha$-anchors, $(\ell-1)$ many $A_V$-choice and~$A^i_E$-choice gadgets for each~$i \in [\ell]$, and $2\binom{\ell}{2}$ many $[m+1]$-choice gadgets.
Recall that a $A$-choice contains exactly~$\max_{a \in A}\{a\}$ vertices.
Thus, the number of vertices in~$H$ is
\[2(\ell-1)\ell \cdot \alpha + (\ell-1)\ell \cdot (2\beta n/\ell + m + 1) + (\ell-1)\ell \cdot (m+1) = 2(\ell-1)\ell \cdot (\alpha + \beta n/\ell + m + 1) = 2(\ell-1)\ell \cdot h.\]
Hence, by pigeon principle, each connected component in~$H-F$ contains exactly~$h$ vertices.
Thus, $F$ certifies also a solution to the \textsc{Equitable Connected Partition} instance~$(H,2(\ell-1)\ell)$.
The correctness of the reduction follows from the arguments of Enciso et al.~\cite{EFGK2009}, as they use the same reduction for \textsc{Equitable Connected Partition}.
\end{proof}

We remark that the above theorem also implies W[1]-hardness with respect to the treewidth for \densestEdgeDeletion{} and $T_{h+1}$-\textsc{Free Edge Deletion} as the treewidth of is upper bounded by twice the feedback edge number.
This resolves an open question of Enright and Meeks~\cite{EM2018}.

\section{\densestVertexDeletion{}}\label{Sec:Critical_k_vertex_deletion}

In this section we show that \densestVertexDeletion{} is \NP-complete on several graph classes. We also prove the W[2]-hardness with respect to the parameter $k$, the number of vertices to remove.

\subsection{Polynomial-time algorithm for trees}\label{sec:vertex-del-poly-time}
In this part we discuss the case of trees and show that the problem is polynomial-time solvable. 
As in the edge deletion variant, we only need to consider the case that~$\tau_\rho < 1$ (see the discussion before \cref{thm:EdgeDel-poly-on-trees}).
Thus, we need to delete vertices such that each connected component in the resulting graph has at most~$h = 1/(1-\tau_\rho)$ many vertices.

Shen and Smith~\cite{SS2012} consider the problem of deleting $k$ vertices to minimize the size of the largest connected component. 
They established an algorithm in $O(n^3 \log n)$ time for trees.  
We improve on this as follows.

\begin{theorem}\label{thm:vertexDel-Trees}
	\densestVertexDeletion{} can be solved in~$O(n)$ time on trees.
\end{theorem}

\begin{proof}
	Let~$T$ be the input tree.
	We propose a simple greedy algorithm that works bottom up from the leaves.
	In each vertex we visit we store the size (number of vertices) of the current subtree.
	For a vertex~$v$ this can be easily computed by summing up the values in its children and adding one.
	Whenever a vertex reaches subtree size above~$h$, we delete the vertex. 
	Naturally, when computing the subtree size, then deleted children will be ignored.
	This algorithm runs in~$O(n)$ time.
	
	Clearly, the algorithm provides a graph where each connected component contains at most~$h$ vertices.
	It remains to show that there is no smaller solution.
	To this end, consider a subtree~$T_v$ rooted at vertex~$v$.
	If~$T_v$ contains more than~$h$ vertices, then any solution must delete at least one vertex in~$T_v$.
	If $T_v$ has size less than~$h$, then we claim there is an optimal solution not deleting any vertex in~$T_v$:
	Assume otherwise and let~$S \subseteq V(T)$ be an optimal solution deleting a vertex~$u$ in~$T_v$.
	Then removing~$u$ from~$S$ and adding the parent from~$v$ to~$S$ results in another solution of the same cost---a contradiction.
	Hence, our algorithm produces an optimal solution.
\end{proof}

\subsection{NP-hardness results}\label{sec:vertex-del-NP-hardness}
We prove in the following the \NP-hardness of \densestVertexDeletion{} by reduction from \textsc{Feedback Vertex Set}, which remains \NP-hard on Hamiltonian planar 4-regular graphs~\cite{CF21}.
Given an instance of \textsc{Feedback Vertex Set}, that is a graph $G$, and an integer $k$, the problem consists in deciding the existence of a subset  $V'\subseteq V(G)$ with $|V'|\leq k$ such that the $G- V'$ is cycle-free.

\densestVertexDeletion{}  is closely related to \textsc{Feedback Vertex Set}.  
Given a graph $G$ and a subset $V'\subseteq V(G)$, we have that $G- V'$ is cycle-free if and only if $\rho^*(G-V')<1$.  
Thus \densestVertexDeletion{} is NP-complete on all classes of graphs where \textsc{Feedback Vertex Set} is NP-complete, for example on planar and maximum degree 4 graphs.
We can even prove a stronger result as follows:

Let $G$ be an undirected graph. 
The bipartite incidence graph of $G$ (also called subdivision of~$G$) is the bipartite graph $H$ whose vertex set is $V(G) \cup E(G)$ and there is an edge in $H$ between $v \in V(G)$ and $e \in  E(G)$ if and only if $e$ is incident to $v$ in $G$. %

\begin{theorem}\label{thm:VertexDeletion_NP_Hard-bipartite}
\densestVertexDeletion{}  is \NP-complete for $\tau_\rho<1$ even for planar bipartite graphs with maximum degree 4. 
\end{theorem}
\begin{proof}
We prove the \NP-hardness by reduction from \textsc{Feedback Vertex Set}.
Considering a graph $G$ on $n$ vertices and an integer $k$, instance of \textsc{Feedback Vertex Set}, let $H$ be the bipartite incidence graph of $G$. 
Remark that if $G$ is planar of maximum degree 4 then $H$ is still planar with maximum degree 4.  We show in the following that  $G$ contains  a subset  $V'\subseteq V(G)$ with $|V'|\leq k$ such that $G- V'$ is cycle-free if and only if  $H$ contains a subset  $F\subseteq V(G)\cup E(G)$ with $|F|\leq k$ such that the $\rho^*(H-F)\leq 1-1/n^2$.

Note that to any cycle  $v_1,v_2,\ldots, v_i, v_1$ in $G$ it corresponds a cycle $v_1, v_1v_2, v_2,\ldots, v_i,v_iv_1,v_1$ in $H$. Moreover, to any cycle from $H$ where the vertices alternate between vertices from $V(G)$ and $E(G)$, there corresponds a cycle in $G$.
Furthermore, a subgraph $G[V']$ is  cycle-free if and only if $\rho(G[V'])< 1$.

If  $G$ contains a subset  $F\subseteq V$ with $|F|\leq k$ such that the $G- F$ is cycle-free, then $H-F$ contains no cycle and thus $\rho^*(H-F)\leq 1-1/n^2$.

Consider, now, that $H$ contains a subset  $F\subseteq V(G)\cup E(G)$ with $|F|\leq k$ such that $\rho^*(H-F)\leq 1-1/n^2$. Any vertex $e \in F\cap E(G)$ from $H$ has two neighbors. We can replace it by one of its neighbors in $F$. This exchange makes the degree of $e$ less than or equal to $1$ in $H-F$, and therefore guarantees that $e$ is not in any cycle. Thus, exchanging $e$ in $F$ for one of its neighbors cannot create any cycles, and moreover $|F|\leq k$. 
At the end of these exchanges, $F$ contains only vertices from $V(G)$ and moreover $H-F$ contains no cycle. Thus $F$ is a solution for the instance $(G,k)$ of \textsc{Feedback Vertex Set}. 
\end{proof}

In the next reduction we use the following problem:
 
\begin{problem}
    \problemtitle{\textsc{Dissociation Set}}
    \probleminput{A graph $G$ and an integer $k$}
    \problemquestion{Is there a subset of vertices $S\subseteq V(G)$ of size at least $k$  such that $G[S]$ has maximum degree at most $1$?}    
\end{problem}
\textsc{Dissociation Set} is \NP-hard even in line graphs of planar bipartite graphs \cite{ODFG2011}.

\begin{theorem}\label{thm:NP-harness-rho12}
		\densestVertexDeletion{} is \NP-complete for $\tau_\rho=1/2$ even for line graphs of planar bipartite graphs.
\end{theorem} 
\begin{proof} 
	Given an instance $(G,k')$ of  \textsc{Dissociation Set} on planar line graphs of planar bipartite graphs we construct the instance $(G,k = n-k',\tau_\rho=1/2)$ of \densestVertexDeletion{}. 
	We can easily see that $G$ has a solution $S$ of size at least $k'$ if and only if $\rho^*(G[S]) \leq 1/2$ as $G[S]$ has density at most $1/2$ and $V(G)\backslash S$ has size at most $k$.
\end{proof}

The two previous results show hardness for small values of~$\tau_\rho$.
Next we show that \densestVertexDeletion{} remains NP-hard for most values of~$\tau_\rho$.

\begin{theorem}
	\label{thm:NP-hardness-any-rho}
	\densestVertexDeletion{} is \NP-complete for any rational $\tau_\rho$  such that $0\leq \tau_\rho \leq n^{1-1/c}$, where $c$ is any constant.
\end{theorem} 
\appendixproof{thm:NP-hardness-any-rho}
{
\begin{proof} 
	Given an instance $(G,k)$ of  \textsc{Vertex Cover} we construct an instance $(G',k,\tau_\rho)$ of \densestVertexDeletion{} as follows. 
	Consider a balanced graph $H$ on $n_H$ vertices and $m_H$ edges such that $ \tau_\rho= m_H/n_H $ and let $u$ be a vertex from $H$. 
	A graph that is  strongly balanced  and thus also balanced of density $ \tau_\rho$ can be constructed in polynomial time  \cite{RV1986}.
	Create $n_G$ copies of it $H_1,\ldots, H_{n_G}$ and identify the vertex $u$ of $H_i$ with a vertex $v_i$ of $G$ for all $i \in \{1,\ldots, n_G\}$.
	Denote this graph by $G'$.
	The instance $(G',k, \tau_\rho)$ of \densestVertexDeletion{} problem is the desired one. 
	We claim that $G$ has a vertex cover $C$ of size at most $k$ if and only if there is a subset $F \subseteq V(G')$ with $|F|\leq k$ such that $\rho^*(G'-F) \leq \tau_\rho$.

	 ``$\Rightarrow:$''  If $C$ is a vertex cover of size at most $k$ then consider $F$ the   vertices  $v_i \in C$. Thus $|F|\leq  k$. Since $C$ is a vertex cover in $G$, the graph $G'-F$ contains as connected components copies of $H$ or copies of  $H-\{u\}$ that are of density  $\tau_\rho$ or at most  $\tau_\rho$ since $H$ is a balanced graph. 
	
	``$\Leftarrow:$''  Let $F$ be a set of at most $k$ vertices from $G'$ such that $\rho^*(G'-F) \leq \tau_\rho$.
	Define $C$ as the set of vertices  $v_i$ such that there exists a vertex $x\in F\cap H_i$. Suppose by contradiction that $C$ is not a vertex cover in $G$, then there are two vertices $v_i,v_j \notin V(G) \setminus C$ such that $v_iv_j\in E(G)$. This corresponds to a subgraph in $G'$ containing the two copies $H_i,H_j$ and edge $\{v_i,v_j\}$ of density $(2m_H+1)/(2n_H) > m_H/n_H=\tau_\rho$. 
	The maximum value $\tau_\rho$ could have is in the case where $n_H = n_G^c$ for some constant $c$ and $H$ is the complete graph $H=K_{m_H}$ and then $G'$ has $n_{G'}=n_H\cdot n_G$ vertices, so $n_G= n_{G'}^{1/(c+1)}$ and $n_H = n_{G'}^{c/(c+1)} = n_{G'}^{1 - 1/(c+1)} = n_{G'}^{1 - 1/c'}$ for~$c' = c - 1$. 
\end{proof}
}

The last result in this subsection concerns split graphs.
The following hardness result holds for large (non-constant) values of~$\tau_\rho$.
\begin{theorem}
	\label{thm:vertex-del-split-graph}
	\densestVertexDeletion{} remains NP-hard on split graphs.
\end{theorem}
\appendixproof{thm:vertex-del-split-graph}
{
\begin{proof}
	We reduce from the NP-hard \textsc{3-Set Cover} \cite{GJ79}.
	\begin{problem}
		\problemtitle{\textsc{3-Set Cover}}
		\probleminput{A set family~$\mathcal{F} = \{S_1, \ldots, S_h\}$ over a universe~$U$ with~$|S_i| = 3$ for all~$i\in [h]$ and an integer~$\ell$.}
		\problemquestion{Is there a size-$\ell$ set cover, that is, a set~$\mathcal{F}' \subseteq \mathcal{F}$ with~$|\mathcal{F}'| \le \ell$ and~$\bigcup_{S \in \mathcal{F}'}=U$?}
	\end{problem}

	Given an instance $I=(U,\mathcal{F},\ell)$ of \textsc{3-Set Cover}, we construct an instance $(G, k, \tau_\rho)$ of \densestVertexDeletion{} as follows:
	Add a set-vertex for every set~$S_i \in \mathcal{F}$ and an element-vertex for every element~$u_j \in U$. 
	Make the set-vertices a clique, add $\ell$ dummy vertices to the clique, and add an edge between a set-vertex and element-vertex if the element is \emph{not} in the set.
	(The dummy vertices are not adjacent to any element-vertex.)
	Set~$\tau_\rho = \ell - (1/2)$ and~$k = h - \ell$.
	The element-vertices form an independent set, thus, the constructed graph is indeed a split graph. We claim that there exists a set~$\mathcal{F}' \subseteq \mathcal{F}$ with~$|\mathcal{F}'| \le \ell$ and~$\bigcup_{S \in \mathcal{F}'}=U$ if and only if there exits a subset $S$ of vertices from $G$, $|S|\leq k$ such that $\rho^*(G-S)\leq \tau_\rho$.
	
	$\Rightarrow$: Let~$\mathcal{F}'$ be a set cover of size exactly~$\ell$ (if the size is less than~$\ell$ add arbitrary sets until the size is~$\ell$).
	We claim deleting the vertices~$S$ corresponding to~$\mathcal{F} \setminus \mathcal{F}'$ is a solution:
	Clearly, $|S| = h - \ell$.
	Each element-vertex has, by construction, degree at most~$\ell - 1$.
	Thus, by \cref{rr:low-degree} we can remove them.
	What remains is a clique on~$2\ell$ vertices of density~$\ell-(1/2) = \tau_\rho$. 
	
	$\Leftarrow$: Let~$S$ be a solution of the instance $(G, k, \tau_\rho)$.
	If~$S$ contains less than~$h-\ell$ vertices from the clique, then the remaining clique has size at least~$2\ell + 1$ and, thus, density~$\ell > \tau_\rho$.
	Thus, $S$ needs to contain exactly~$h-\ell$ vertices in the clique. If $S$ contains any dummy-vertices we can simply exchange them with set-vertices without increasing the density.
	We claim that the remaining~$\ell$ set-vertices~$C$ correspond to a set cover.
	Assume otherwise.
	Then there is an element-vertex~$v$ adjacent to all remaining set-vertices, that is, $\deg(v)=\ell$.
    Since $\ell>\ell-(1/2)$, by \cref{lem:fractions}, the graph induced by the remaining clique together with $v$ has density strictly greater than $\tau_\rho$---a contradiction.
\end{proof}
}

\subsection{Parameterized complexity results}\label{sec:vertex-del-parameterized}

We show that \densestVertexDeletion{} is FPT with respect to vertex cover number and W[2]-hard with respect to $k$, the number of vertices to remove.

\begin{theorem}\label{thm:VertexDeletion_FPT_VC}
	\densestVertexDeletion{} can be solved in time $2^{O(\ell^2)} n^{O(1)}$ where~$\ell$ is the vertex cover number.
\end{theorem}
\begin{proof}
	Let $(G,k,\tau_\rho)$ be an instance of \densestVertexDeletion{} where $G$ has vertex cover number $\ell$. 
	One can find a minimum vertex cover of size $\ell$ in time $O(2^{\ell} + n + m)$ \cite{CFKL2015}. 
	Denote by $C$ the set of vertices that belongs to the minimum vertex cover, and by $S=V(G)-C$ the set of vertices in the independent set. 
	We divide the vertices in $S$ into at most $2^\ell$ classes $I_1, \ldots, I_{2^{\ell}}$, where two vertices $v_1,v_2\in S$ are in the same class $I_i$ if they have the same neighbors in $C$. 

	Notice that for $k\geq \ell$ we always have a yes-instance, as deleting the $\ell$ vertices in the vertex cover yields a graph with density $0$. 
	Thus, we are interested in the case where we delete at most $\ell-1$ vertices, that is $k\leq \ell-1$.
	
	The vertices in each $I_i$ (with $i=1,\ldots, 2^\ell$) are all indistinguishable from each other, and we need to delete between $0$ and $\ell-1$ vertices $S$. 
	Thus we need to check at most $(\ell{2^\ell})^\ell$ sets of vertices from $S$ as candidates for deletion. 
	For vertices in $C$ we check all $2^{\ell}$ subsets of vertices from $C$ as candidates for deletion. 
	In total, we check at most $2^{\ell} (\ell{2^\ell})^\ell$ subsets as candidates for deletion, yielding a running time of $2^{\ell^2 + \ell\log \ell +\ell} n^{O(1)}$.
\end{proof}

\begin{remark}
	Note that if $\tau_\rho \ge \ell$, then nothing needs to be deleted:
	On the one hand, any vertex from the independent set has degree at most $\ell\le\tau_\rho$. By
	\cref{lem:nice-props} (\ref{niceprop-minneighbors-within-densest}.) no vertex of the independent set belongs to a subgraph of density greater than $\tau_\rho$. On the other hand, the vertex cover has density at most $(\ell-1)/2<\tau_\rho$.
\end{remark}

The following hardness result holds for large (non-constant) values of~$\tau_\rho$.
\begin{theorem}
	\label{thm:vertex-del-w2}
	\densestVertexDeletion{} is W[2]-hard with respect to the number~$k$ of vertices to remove.
\end{theorem}
\appendixproof{thm:vertex-del-w2}
{
\begin{proof}
	We provide a parameterized reduction from \textsc{Set Cover}, which is W[2]-complete with respect to the solution size~$\ell$~\cite{CFKL2015}.
	\begin{problem}
		\problemtitle{\textsc{Set Cover}}
		\probleminput{A set family~$\mathcal{F} = \{S_1, \ldots, S_q\}$ over a universe~$U$ and an integer~$\ell$.}
		\problemquestion{Is there a size-$\ell$ set cover, that is, a set~$\mathcal{F}' \subseteq \mathcal{F}$ with~$|\mathcal{F}'| \le \ell$ and~$\bigcup_{S \in \mathcal{F}'}=U$?}
	\end{problem}
	
	Given an  instance~$(\mathcal{F},U,\ell)$ of \textsc{Set Cover}, we construct an instance~$(G, k = \ell, \tau_\rho)$ of \densestVertexDeletion{} as follows:
	We assume without loss of generality that~$|\mathcal{F}|$ is even (if not, duplicate any set). Graph $G$ is constructed as follows.
	For each element~$u \in U$ add a vertex~$w_u$.
	Let~$V_U = \{w_u \mid u \in U\}$ be the set of all ``element vertices''. 
	Next, for each set~$S \in \mathcal{F}$ add a vertex~$v_S$ and add for each $u \in U$ the edge~$\{v_S, w_u\}$ if $u \in S$. Call the vertices of all ``set vertices'' $V_{\mathcal{F}} = \{v_S \mid S \in \mathcal{F}\}$.
		Set~$h = 2(|U| + |\mathcal{F}|)$ and~$\tau_\rho = h/2$. 
	Add two sets $V_A$ and~$V_B$ of vertices. 
	The set~$V_A$   contains $h-|\mathcal{F}|+1+\ell$ vertices.
	Make~$V_A \cup V_{\mathcal{F}}$ a clique, that is, add all edges between the vertices in~$V_A \cup V_{\mathcal{F}}$. 
	Thus, $G[V_A \cup V_{\mathcal{F}}]$ is~$(h + \ell)$-regular.
	The set~$V_B$ is a clique of size~$h+1$.
	Thus, $G[V_B]$ is~$h$-regular.
	Finally, make each vertex in~$V_U$ adjacent to arbitrary vertices in~$V_B$ so that each vertex in~$V_U$ has exactly~$h/2 + 1$ neighbors in~$G$ (recall that~$h$ is even).
	This finishes the construction of $G$, see \cref{fig:w[2]-hard-vertex-removal-example} for an illustration. 
	
	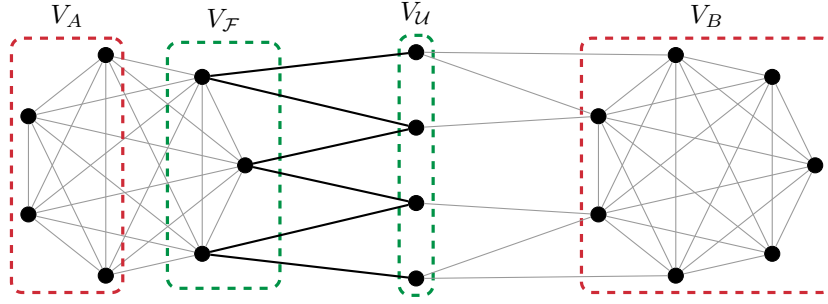
\begin{figure}[t]
		\centering
		\begin{tikzpicture}[scale = 1]
			\def\k{7}
			\def\radius{1.5}
			\pgfmathtruncatemacro{\kk}{\k - 1}
			\foreach \i in {1,...,\k}
			{
				\node[circle,draw,fill,inner sep=2pt] at ({\i * 360 / \k - 2 * 360 / \k}:\radius) (v-\i) {};
				\node[circle,draw,fill,inner sep=2pt,xshift=5*\radius cm] at ({\i * 360 / \k - 2 * 360 / \k}: \radius) (w-\i) {};
			}
			\foreach[count=\xi from 2] \i in {1,...,\kk}
			{
				\foreach \j in {\xi,...,\k}
				{
					\draw[black!40] (v-\i) to (v-\j);
					\draw[black!40] (w-\i) to (w-\j);
				}
			}
			\node[inner sep = 10pt,very thick, italyGreen, draw,rounded corners, dashed,fit=(v-1) (v-2) (v-3),label=above:$V_{\mathcal{F}}$] {};
			\node[very thick, italyRed, draw,rounded corners, dashed,fit=(v-4) (v-5) (v-6) (v-7),label=above:$V_{A}$] {};
			\node[very thick, italyRed, draw,rounded corners, dashed,fit=(w-1) (w-2) (w-3) (w-4) (w-5) (w-6) (w-7),label=above:$V_{B}$] {};
			
			\foreach \i in {1,...,4}
			{
				\node[circle,draw,fill,inner sep=2pt] at (2.5 * \radius,\i - 2.5) (u-\i) {};
			}
			\node[very thick, italyGreen, draw,rounded corners, dashed,fit=(u-1) (u-4),label=above:$V_{\mathcal{U}}$] {};
			
			\foreach \x / \y in {1/1, 1/2, 2/2, 2/3, 3/3, 3/4}
			{
				\draw[thick] (v-\x) to (u-\y);
			}
			\foreach \x / \y in {1/6, 1/7, 2/6, 3/5, 4/5, 4/4}
			{
				\draw[black!40] (u-\x) to (w-\y);
			}
		\end{tikzpicture}
		\caption{
			A schematic illustration of the reduction in \cref{thm:vertex-del-w2}. 
			The \textsc{Set Cover} instance is encoded in the displayed graph with the edges between~$V_{\mathcal{F}}$ and~$V_U$ encoding the sets and elements.
			The two red marked vertex sets~$V_A$ and~$V_B$ are dummy vertices to ensure certain degrees of the vertices in~$V_{\mathcal{F}}$ and~$V_U$.
		}
		\label{fig:w[2]-hard-vertex-removal-example}
	\end{figure}
	
	Clearly, the reduction runs in polynomial-time.
	It thus remains to prove the correctness, that is, show that~$(\mathcal{F},U,\ell)$ is a yes-instance of \textsc{Set Cover} if and only if~$(G, k = \ell, \tau_\rho)$ is a yes-instance of \densestVertexDeletion{}.
	
	``$\Rightarrow:$''
	Let~$\mathcal{F}'$ be a size-$\ell$ set cover. 
	We claim that deleting~$V_{\mathcal{F}'} = \{v_S \mid S \in \mathcal{F}'\}$ results in a graph with~$\rho^*(G-V_{\mathcal{F}'}) \le h/2$.
	Note that~$G - V_{\mathcal{F}'} - V_U$ is~$h$-regular and thus its densest subgraph has density~$h/2$ (to see this recall that the density of a graph is half its average degree).
	Moreover, the set~$V_U$ is an independent set with each vertex having degree at most~$\tau_\rho = h/2$ in~$G - V_{\mathcal{F}'}$ (as~$\mathcal{F}'$ is a set cover).
	Hence, it follows from \cref{lem:fractions} that adding any subset of~$V_U$ to~$G - V_{\mathcal{F}'}$ will not change~$\rho^*$.
	
	``$\Leftarrow:$'' 
	Let~$S$ be the set of vertices to remove with~$|S| \le k = \ell$. 
	Note that~$G[V_A \cup V_{\mathcal{F}}]$ is a clique on~$h + \ell + 1$ vertices.
	Thus~$S \subseteq V_A \cup V_{\mathcal{F}}$ and~$|S| = \ell$ as otherwise~$G[(V_A \cup V_{\mathcal{F}}) \setminus S]$ has density more than~$h/2$.
	Hence, no vertex in~$V_U$ or $V_B$ can be deleted.
	Set~$G' = G[(V_A \cup V_{\mathcal{F}} \cup V_B) \setminus S]$.
	Note that~$G'$ is~$h$-regular and thus of density exactly~$h/2$. 
	We claim~$V_{\mathcal{F'}} = V_{\mathcal{F}} \cap S$ encodes a set cover of size at most~$\ell$.
	Clearly~$|V_{\mathcal{F'}}| \le \ell$ as~$|S| = \ell$.
	Assume towards a contradiction that~$V_{\mathcal{F'}}$ does not encode a set cover, that is, there is an element~$u \in U$ such that~$u \notin \bigcup_{v_S \in V_{\mathcal{F'}}} S$.
	Thus, $v_u$ has~$h/2+1$ neighbors in~$G'$.
	This contradicts \cref{lem:nice-props} (\ref{niceprop-maxneighbors-in-densest}.).
\end{proof}
}

\section{Conclusion}
Our work provides a first (parameterized) analysis of \densestEdgeDeletion{} and \densestVertexDeletion.
Both problems turn out to be NP-hard even in restricted cases but polynomial-time solvable on trees and cliques.
While the W[1]-hardness for \densestEdgeDeletion{} with respect to the feedback edge number rules out fixed-parameter tractability with respect to many other parameterizations, the respective reduction relies on a very specific target density.
In fact, \cref{fig:density-edge-deletion} seems to indicate that the computational complexity of \densestEdgeDeletion{} might change when altering the target density~$\tau_\rho$ a bit.
Does this alternating pattern between tractable and intractable target density extend beyond target density 2?
For example the polynomial-time solvable $f$-\textsc{Factor} problem might be helpful in designing polynomial-time algorithms for \densestEdgeDeletion{} with~$\tau_\rho \in \NN$ (any $2\tau_\rho$-regular subgraph would be a valid resulting graph~$G-F$).
Could such behavior be exploited in approximation algorithms?

We leave also several open questions concerning the parameterized complexity of these problems.
Is \densestVertexDeletion{} fixed-parameter tractable with respect to the treewidth (plus the solution size)?
What is the parameterized complexity with respect to parameters that are smaller than the vertex cover number, e.\,g., vertex integrity or twin cover number?

\bibliography{cleandensest.bib}

\end{document}